\documentclass[11pt]{article}

\usepackage[letterpaper,margin=1in]{geometry}
\usepackage[pdfpagelabels,bookmarks=true]{hyperref}

\usepackage{amsmath, amssymb}
\usepackage{amsthm}

\usepackage{standalone}

\hypersetup{colorlinks, linkcolor=darkblue, citecolor=darkgreen, urlcolor=darkblue}
\newcommand{\footremember}[2]{\footnote{#2}
    \newcounter{#1}
    \setcounter{#1}{\value{footnote}}}
\newcommand{\footrecall}[1]{\footnotemark[\value{#1}]}

\usepackage{anyfontsize}
\usepackage{complexity}

\usepackage{xspace}

\usepackage[inline]{enumitem}
\setlist[enumerate,1]{label=(\roman*), leftmargin=2.2em}
\setlist[enumerate,2]{label=(\alph*)}
\setlist{nosep,topsep=0.1em}
\setlist[itemize,1]{label={\bfseries--}}

\usepackage{mathtools}

\makeatletter
\newtheorem*{rep@theorem}{\rep@title}
\newcommand{\newreptheorem}[2]{\newenvironment{rep#1}[1]{\def\rep@title{\cref{##1}}\begin{rep@theorem}}{\end{rep@theorem}}}
\makeatother

\newreptheorem{theorem}{Theorem}
\newreptheorem{lemma}{Lemma}
\newreptheorem{corollary}{Corollary}
\newreptheorem{observation}{Observation}
\newreptheorem{remark}{Remark}

\usepackage{mdframed}

\usepackage{xcolor}
\definecolor{darkblue}{rgb}{0,0,0.38}
\definecolor{darkred}{rgb}{0.6,0,0}
\definecolor{darkgreen}{rgb}{0.1,0.35,0}

\usepackage[
backend=biber,
style=alphabetic,
citestyle=alphabetic,
maxalphanames=6,
maxcitenames=99,
mincitenames=98,
maxbibnames=99,
giveninits=true,
url=false,
doi=true,
isbn=true,
backref=true
]{biblatex}

\usepackage{xpatch}

\makeatletter
\patchcmd\blx@bblinput{\blx@blxinit}
                      {\blx@blxinit
                      }{}{\fail}
\makeatother

\makeatletter
\DeclareFieldFormat{eprint:arxiv}{arXiv\addcolon\space
  \ifhyperref
    {\href{https://arxiv.org/abs/#1}{\nolinkurl{#1}\iffieldundef{eprintclass}
     {}
{\addspace\mkbibbrackets{\thefield{eprintclass}}}}}
    {\nolinkurl{#1}
     \iffieldundef{eprintclass}
       {}
{\addspace\mkbibbrackets{\thefield{eprintclass}}}}}
\makeatother

\usepackage[vlined,ruled,algo2e,linesnumbered]{algorithm2e}
\SetAlCapHSkip{0.2em}
\SetAlgoInsideSkip{smallskip}
\SetCustomAlgoRuledWidth{0.95\textwidth}
\makeatletter
\patchcmd{\@algocf@start}{\begin{lrbox}{\algocf@algobox}}{\rule{0.025\textwidth}{\z@}\begin{lrbox}{\algocf@algobox}\begin{minipage}{0.95\textwidth}}{}{}
\patchcmd{\@algocf@finish}{\end{lrbox}}{\end{minipage}\end{lrbox}}{}{}
\makeatother

\usepackage[margin=10pt,font=small,labelfont=bf]{caption}
\usepackage[margin=10pt,font=small,labelfont=bf]{subcaption}

\usepackage{graphicx}
\graphicspath{{.}{graphics/}}
\makeatletter
\newcommand\appendtographicspath[1]{\g@addto@macro\Ginput@path{#1}}
\makeatother

\usepackage{tikz}
\usepackage{pgfplots}
\usetikzlibrary{calc}
\usetikzlibrary{positioning}
\usetikzlibrary{decorations.pathmorphing}
\usetikzlibrary{decorations.pathreplacing}
\usetikzlibrary{decorations.markings}
\usetikzlibrary{backgrounds}
\usetikzlibrary{shapes.misc, shapes.geometric}

\makeatletter
\DeclareRobustCommand{\cev}[1]{{\mathpalette\do@cev{#1}}}
\newcommand{\do@cev}[2]{\vbox{\offinterlineskip
    \sbox\z@{$\m@th#1 x$}\ialign{##\cr
      \hidewidth\reflectbox{$\m@th#1\vec{}\mkern4mu$}\hidewidth\cr
      \noalign{\kern-\ht\z@}
      $\m@th#1#2$\cr
    }}}
\makeatother

\usepackage[capitalize, nameinlink]{cleveref}

\newtheorem*{theorem*}{Theorem}
\newtheorem{theorem}{Theorem}
\newtheorem{lemma}[theorem]{Lemma}
\newtheorem*{lemma*}{Lemma}
\newtheorem{conjecture}[theorem]{Conjecture}

\newtheorem{definition}[theorem]{Definition}
\newtheorem{remark}[theorem]{Remark}
\newtheorem{corollary}[theorem]{Corollary}

\newtheorem{observation}[theorem]{Observation}

\newtheorem{claim}[theorem]{Claim}

\AddToHook{env/conjecture/begin}{\crefalias{theorem}{conjecture}}
\AddToHook{env/lemma/begin}{\crefalias{theorem}{lemma}}
\AddToHook{env/claim/begin}{\crefalias{theorem}{claim}}
\AddToHook{env/fact/begin}{\crefalias{theorem}{fact}}
\AddToHook{env/remark/begin}{\crefalias{theorem}{remark}}
\AddToHook{env/observation/begin}{\crefalias{theorem}{observation}}
\AddToHook{env/line/begin}{\crefalias{theorem}{line}}
\AddToHook{env/figure/begin}{\crefalias{theorem}{figure}}
\AddToHook{env/corollary/begin}{\crefalias{theorem}{corollary}}

\crefname{theorem}{Theorem}{Theorems}
\crefname{conjecture}{Conjecture}{Conjectures}
\Crefname{lemma}{Lemma}{Lemmas}
\Crefname{claim}{Claim}{Claims}
\Crefname{fact}{Fact}{Facts}
\Crefname{remark}{Remark}{Remarks}
\Crefname{observation}{Observation}{Observations}
\Crefname{line}{Line}{Lines}
\Crefname{figure}{Figure}{Figures}
\Crefname{corollary}{Corollary}{Corollaries}
\crefalias{AlgoLine}{line}

\usepackage{multicol}
 \usepackage{todonotes}

\title{Short circuit walks in fixed dimension}

\author{
Alexander E. Black
\footremember{bowdoin}{
Bowdoin College, Brunswick, ME, USA.
    Email: \href{mailto:a.black@bowdoin.edu}{a.black@bowdoin.edu}
}
\and
Christian N{\"o}bel\footremember{ETH}{
ETH Zurich, Zurich, Switzerland. R.S. funded by SNSF Ambizione Grant No. 216071. 
Email: $\{$\href{mailto:cnoebel@ethz.ch}{cnoebel}, \href{mailto:rsteine@ethz.ch}{rsteine}$\}$@ethz.ch.}\and
Raphael Steiner\footrecall{ETH}}

\date{}
\renewcommand{\epsilon}{\varepsilon}
\begin{document}

\maketitle
\thispagestyle{empty}
\addtocounter{page}{-1}

\begin{abstract}
Circuit augmentation schemes are a family of combinatorial algorithms for linear programming that generalize the simplex method. To solve the linear program, they construct a so-called \emph{monotone circuit walk}: They start at an initial vertex of the feasible region and traverse a discrete sequence of points on the boundary, while moving along certain allowed directions (circuits) and improving the objective function at each step until reaching an optimum. Since the existence of short circuit walks has been conjectured (\emph{Circuit Diameter Conjecture}), several works have investigated how well one can efficiently approximate shortest monotone circuit walks towards an optimum. A first result addressing this question was given by De Loera, Kafer, and Sanit\`{a} [SIAM J.~Opt., 2022], who showed that given as input an LP and the starting vertex, finding a $2$-approximation for this problem is \NP-hard. Cardinal and the third author [Math.~Prog.~2023] gave a stronger lower bound assuming the exponential time hypothesis, showing that even an approximation factor of $O(\frac{\log m}{\log \log m})$ is intractable for LPs defined by $m$ inequalities. Both of these results were based on reductions from highly degenerate polytopes in combinatorial optimization with high dimension. 

In this paper, we significantly strengthen the aforementioned hardness results by showing that for every fixed $\varepsilon>0$ approximating the problem on polygons with $m$ edges to within a factor of $O(m^{1-\varepsilon})$ is \NP-hard. This result is essentially best-possible, as it cannot be improved beyond $o(m)$. In particular, this implies hardness for simple polytopes and in fixed dimension. 
\end{abstract} 
\newpage

\section{Introduction}

Circuit augmentation schemes are a family of combinatorial linear programming algorithms generalizing the simplex method and taking inspiration from interior point methods.
Like the simplex method, circuit augmentation schemes start at an initial vertex of a polytope and follow a sequence of discrete steps along the boundary of the polytope until reaching the optimum.
However, unlike the simplex method and akin to interior point methods, these steps may move along the interior of the polytope and need not move from vertex to vertex.
To be explicit, for a polyhedron of the form $\{\mathbf{x}\in \mathbb{R}^d\colon A\mathbf{x} \leq \mathbf{b}\}$ for $A$ an $n \times d$ matrix, a \emph{circuit direction} is any vector $\mathbf{w}$ parallel to a line given by $\{\mathbf{x}: A_{I} \mathbf{x} = 0\}$, where $A_{I}$ is a $(d-1) \times d$ sub-matrix of $A$ of rank $d-1$. 
Then a \emph{circuit step} is any step from a point $\mathbf{p}$ on the boundary of a polytope $P$ that goes from $\mathbf{p}$ to $\mathbf{p} + \lambda^{\ast} \mathbf{w}$, where $\lambda^{\ast} = \max(\{\lambda \in \mathbb{R}_{\geq 0}: \mathbf{p} + \lambda \mathbf{w} \in P\})$ and $\mathbf{w}$ is some circuit of $P$.
Finally, a \emph{circuit walk} is a path from point to point on the boundary of the polytope consisting of circuit steps. If we are additionally given a linear objective function $\mathbf{c}^\top\mathbf{x}$, a \emph{monotone circuit walk} is any circuit walk which improves the value of the objective function at each step. See \cref{fig:circuit_walk_example} for an example. We remark that while all directions parallel to edges of a polytope are also circuits, the converse is generally not true. However, circuit directions do coincide with edge directions in the case of polygons, as later formally stated in \cref{obs:edgesarecircuits}.

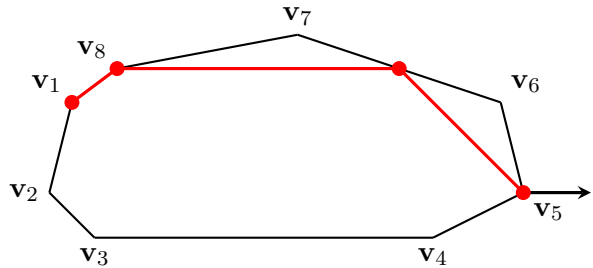
\begin{figure}[h!]
    \centering
    \begin{tikzpicture}[scale=0.3]

\begin{scope}
\coordinate (1) at (7,-5) {};
\coordinate (2) at (6,-9.00) {};
\coordinate (3) at (8,-11) {};
\coordinate (4) at (23,-11) {};
\coordinate (5) at (27,-9.0) {};
\coordinate (6) at (26,-5) {};
\coordinate (7) at (17,-2) {};
\coordinate (8) at (9.0,-3.5) {};
\end{scope}

\coordinate (inter) at (21.5,-3.5);

\begin{scope}

\node[above left] at (1) {$\mathbf{v}_1$};
\node[left] at (2) {$\mathbf{v}_2$};
\node[below] at (3) {$\mathbf{v}_3$};
\node[below] at (4) {$\mathbf{v}_4$};
\node[below right] at (5) {$\mathbf{v}_5$};
\node[above right] at (6) {$\mathbf{v}_6$};
\node[above] at (7) {$\mathbf{v}_7$};
\node[above left] at (8) {$\mathbf{v}_8$};

\end{scope}

\draw[very thick, -stealth] (5) -- ($(5) + (3,0)$);

\begin{scope}[thick]
\draw (1) -- (2);
\draw (2) -- (3);
\draw (3) -- (4);
\draw (4) -- (5);
\draw (5) -- (6);
\draw (6) -- (7);
\draw (7) -- (8);
\draw (8) -- (1);
\end{scope}

\begin{scope}[every node/.style={fill=red, minimum size=3,circle, inner sep=2pt}]
\node at (1){};
\node at (5){};
\node at (8){};
\node at (inter) {};
    
\end{scope}

\begin{scope}[very thick, red]
\draw (1) -- (8);
\draw (8) -- (inter);
\draw (inter) -- (5);

\end{scope}

\end{tikzpicture}
     \caption{Example of a monotone circuit walk from $\mathbf{v}_1$ to $\mathbf{v}_5$ for $\mathbf{c} = (1,0)^\top$.
    The circuit walk uses the directions of the edges $\{\mathbf{v}_1,\mathbf{v}_8\}$, $\{\mathbf{v}_3,\mathbf{v}_4\}$, and $\{\mathbf{v}_2,\mathbf{v}_3\}$.
    Note that going from $\mathbf{v}_1$ to $\mathbf{v}_2$ and then to $\mathbf{v}_5$ would give a shorter, but not $\mathbf{c}$-monotone, circuit walk.
    }
    \label{fig:circuit_walk_example}
\end{figure}

The choice of circuit walk for solving a linear program is not canonical and a method of choosing such a walk is called a pivot rule \cite{CircuitAugmentaionOrig}.
There are pivot rules guaranteeing a weakly polynomial run-time for circuit augmentation schemes \cite{SeanThesis, CircuitPivotRules, CircuitImbalanceBounds}.
Making circuit augmentation schemes competitive in practice is an ongoing effort in the community with recent progress made especially by Borgwardt and Viss \cite{PolyhedralModel, SteepestDescentAugmentation, VissThesis}.
Furthermore, circuit augmentation is a useful tool for analyzing both the simplex method \cite{01simplex} and interior point methods   \cite{WorsethanSimplex, StronglyPolyGenFlow}.

In contrast to self concordant barrier interior point methods for which no strongly polynomial time version is possible \cite{IPMNotPolyFirst, IPMNotPolyStrong}, circuit augmentation schemes remain a candidate solution to Smale's 9th problem \cite{Smale1998} asking for the existence of a strongly polynomial time algorithm for linear programming.
One challenge for this problem is the difficulty of the polynomial Hirsch conjecture \cite[Conjecture 1.3]{santos} and the analogous circuit diameter conjecture \cite{CircuitDiamConjecture}, which respectively ask whether a polynomial length path in the graph of the polytope and a polynomial length circuit walk always exists between a pair of vertices of the polytope.
However, even if short paths exist, the challenge still remains to provide an efficient pivot rule guaranteed to follow one.
Even for polytopes whose graph is isomorphic to that of a hyper-cube, called combinatorial hyper-cubes, there is no known strongly polynomial time linear programming algorithm for optimizing over them.
In fact, combinatorial cubes are used to construct hard instances for both the simplex method \cite{KleeMinty} and interior point methods \cite{IPMNotPolyStrong}.
Finally, \textcite{RockExtensions} showed that the general linear programming problem can be reduced in strongly polynomial time to linear programs over simple polytopes whose graphs have linear diameter.
Hence, the key question for understanding the complexity of the simplex method and circuit augmentation schemes in relation to Smale's problem is whether one can efficiently compute a short path assuming one exists. 

The computational problem of finding a shortest path between vertices of a polytope is already known to be hard in various guises both in the setting of edge walks along the graph of the polytope and circuit walks.
In the former, the first problem of this type to be studied was the combinatorial diameter, the diameter of the graph of the polytope.
In their 1994 work Frieze and Teng \cite{FriezeTeng1994} showed that computing the combinatorial diameter of a polytope is \NP-hard.
This result was much later improved by Sanit\`{a} in 2018 \cite{Sanita18} who showed that the same problem is strongly \NP-hard even for fractional matching polytopes.
For circuit walks, the analogous notion is the circuit diameter, the maximal length of a shortest circuit walk between any pair of the vertices of the polytope.
Computing the circuit diameter was recently shown to be strongly \NP-hard by \textcite{CircDiamNPHard}.
Very recently, both of these results were extended by \textcite{HarderthanNPHard}, who showed that computing the combinatorial diameter and circuit diameter are (conjecturally) harder than problems in \NP{} by showing that they are in fact $\Pi_{2}$-complete. 
However, asking for the combinatorial diameter or circuit diameter of a polytope is distinct from asking for an algorithm to find a short path or circuit walk towards an optimum, which is all that is needed for linear programming algorithms.
Finding a shortest path on graphs of polytopes is known to be \NP-hard \cite[Theorem 2]{CircuitPivotRules}.
It is furthermore \NP-hard even for graphs of several highly structured polytopes, such as alcoved polytopes and classes of generalized permutahedra \cite{GraphicalZonotopes, GraphAssociahedra, PolymatroidShortPaths}.
For circuit walks, it is known \cite[Corollary~1]{CircuitPivotRules} that even checking adjacency with the optimum of a linear program is \NP-hard, implying that $(2-\varepsilon)$-approximating shortest (monotone) circuit walks to an optimum is intractable.
Similar results are known for highly structured families of polytopes from combinatorial optimization \cite{ShortSignCompatible, ReconfigAlternatingCycles}. The strongest inapproximability results currently available, both for shortest paths in graphs of polytopes and for shortest circuit walks, are due to Cardinal and the third author~\cite[Theorem~1, Corollary~2]{cardinal_steiner_23} who showed that no polynomial-time algorithm can approximate shortest (monotone) paths or circuit walks to an optimum to within any constant factor (assuming $\P\neq \NP$) or to within a factor $O\left(\frac{\log m}{\log \log m}\right)$ (assuming the exponential time hypothesis), where $m$ is the number of inequalities in the input polytope description. 

In fixed dimension, it was already noted by Frieze and Teng \cite{FriezeTeng1994} that one can always trivially find a shortest path between any pair of vertices in the graph of the polytope in polynomial time.
For polygons\footnote{Throughout this paper, when speaking of polygons we always mean filled convex polygons.}, this is especially easy to see as there are only two paths to choose from.
However, such an observation has not been made for circuit augmentation schemes.
Can one find a shortest improving circuit walk to the optimum in polynomial time, at least if the dimension is fixed as a constant?
The following problem yields the formal setup to study this question.
\begin{mdframed}[innerleftmargin=0.5em, innertopmargin=0.5em, innerrightmargin=0.5em, innerbottommargin=0.5em, userdefinedwidth=0.95\linewidth, align=center]
	{\textsc{Monotone Circuit Distance}}
	\sloppy

	\noindent
	\textbf{Input:} A polytope $P = \{\mathbf{x}\in \mathbb{R}^d\colon Ax\leq \mathbf{b}\}$ defined by a matrix $A\in \mathbb{Q}^{m\times d}$ and a vector $\mathbf{b}\in \mathbb{Q}^m$, a vertex $\mathbf{s}$ of $P$, a cost vector $\mathbf{c}\in \mathbb{Q}^d$, and $k\in \mathbb{Z}_{\geq 0}$.

	\noindent
	\textbf{Decision:} Is there a $\mathbf{c}$-monotone circuit walk from $\mathbf{s}$ to a $\mathbf{c}$-maximal vertex of $P$ of length at most $k$?
\end{mdframed}
\subsection*{Our results.}
\paragraph*{Hardness in fixed dimension.} Our first main result in this paper answers the above question negatively, showing that (perhaps surprisingly) \textsc{Monotone Circuit Distance} is \NP-hard in fixed dimension, already for $d=2$.

\begin{theorem}\label{thm:NP-hard-polygons}
\textsc{Monotone Circuit Distance} is \NP-hard for  polygons.    
\end{theorem}

\paragraph*{Approximation hardness.} In fact, we obtain the following much stronger version of \cref{thm:approx-hardness-monotone-distance}.

\begin{theorem}\label{thm:approx-hardness-monotone-distance}
For every fixed $\varepsilon>0$ it is \NP-hard to solve \textsc{Monotone Circuit Distance} restricted to inputs $(P,\mathbf{s},\mathbf{c},k)$ with the following properties: $d=2$, $P$ is a polygon with $m$ edges, and either there exists a $\mathbf{c}$-monotone circuit walk from $\mathbf{s}$ to a $\mathbf{c}$-maximal vertex of length at most $k$, or there is no such walk of length at most $m^{1-\varepsilon}\cdot k$. 
\end{theorem}
A direct consequence of \cref{thm:approx-hardness-monotone-distance} is a significant improvement of the best known approximability lower bound for \textsc{Monotone Circuit Distance} from $\Omega\left(\frac{\log m}{\log \log m}\right)$ (\cite{cardinal_steiner_23}) to $m^{1-o(1)}$:

\begin{corollary}\label{cor:inapprox}
For every $\varepsilon>0$ the following is \NP-hard: Given as input a polygon $P$ with $m$ edges, a starting vertex $\mathbf{s}$, and a vector $\mathbf{c}\in \mathbb{Q}^2$, compute a monotone circuit walk from $\mathbf{s}$ to a $\mathbf{c}$-maximal vertex of $P$ approximating the minimum possible length of such a walk to within a factor of $m^{1-\varepsilon}$.
\end{corollary}

In particular, this shows that it is \NP-hard to approximate shortest monotone circuit walks to an optimum to within a factor of $m^{1-\varepsilon}$ for linear programs defined by $m$ inequalities.
Interestingly, one can observe that the inapproximability guarantee in \cref{cor:inapprox} is essentially best-possible, as it cannot be improved to any function in $\Omega(m)$: 
\begin{remark}[$\ast$]\label{obs:bruteforce}
    For every constant $K\in \mathbb{N}$, there exists an efficient algorithm that, given as input a polygon $P$ defined by $m$ inequalities, a starting vertex $\mathbf{s}$ of $P$ and a direction $\mathbf{c}\in \mathbb{Q}^2$, computes a $\mathbf{c}$-monotone circuit walk from $\mathbf{s}$ to a $\mathbf{c}$-optimal vertex whose length is at most $\max\{\frac{m}{K}, 1\}$ times the length of a shortest such walk. 
\end{remark}
The simple proof of \cref{obs:bruteforce} can be found\footnote{Throughout this paper, statements whose proofs are deferred to the appendix are marked with a $\ast$-symbol.} in \cref{sec:missing-proofs}.

Summarizing, the above results show that finding monotone shortest circuit walks is meaningfully harder than finding monotone shortest paths in the graph of the polytope.

\paragraph*{The role of degeneracy.} In the theory surrounding the simplex method, one typically assumes without loss of generality that the feasible region forms a \emph{simple} polytope, meaning that each vertex is determined by precisely dimension many inequalities.
Otherwise, the simplex method may follow steps that do not correspond to moving along edges, called \emph{degenerate pivots}. Note that in this case, the number of steps taken in the path can potentially significantly underestimate the run-time, as it does not account for these degenerate pivots.
For circuit augmentation schemes, degeneracy is equally relevant.
Namely, for circuit augmentation algorithms as implemented by \textcite{PolyhedralModel}, one needs to find an initial feasible circuit.
For degenerate polytopes this task requires some computation, but for simple polytopes, one can simply initialize at any improving edge direction from a simplex pivot.
Hence, in both settings, the question of whether one can find shortest paths on simple polytopes in polynomial time is motivated.
In fact, the question of whether computing the diameter of a simple polytope is \NP-hard is asked in the commentary following Problem $10$ of the 2003 survey \cite{KaibelPfetsch2003}.
The computational complexity of finding shortest (monotone) paths in graphs of simple polytopes is also stated explicitly as an open question in the discussion following Theorem $2$ in \cite{CircuitPivotRules}.

In contrast to this, to the best of our knowledge, all previous hardness results for finding shortest paths or circuit walks in polytopes are only for highly degenerate instances, where many inequalities meet at a vertex. Given the above explanation for why simple polytopes should be considered particularly relevant, we would like to emphasize that since polygons are simple polytopes, \cref{thm:NP-hard-polygons} is the first result of its kind that also establishes hardness for simple polytopes.
\begin{corollary}
\textsc{Monotone Circuit Distance} is \NP-hard for simple polytopes. 
\end{corollary}  

\paragraph*{Extending to higher dimensions.}

A priori, our hardness lower bound in \cref{cor:inapprox} could disappear if one considers $d$-dimensional polytopes for $d \geq 3$.
However, using a product of the constructed polygon with a simplex we can lift our results to higher dimensions.
\begin{lemma}[$\ast$]
\label{thm:anyfixeddim}
For every $d\in \mathbb{Z}_{\geq 2}$, given as input a polygon $P\in \mathbb{R}^2$, a vector $\mathbf{c}\in \mathbb{Q}^2$, and a vertex $\mathbf{s}$ of $P$, one can efficiently determine a $d$-dimensional polytope $P_d\in \mathbb{R}^d$, a vector $\mathbf{c}_d\in \mathbb{Q}$, and a vertex $\mathbf{s}_d$ of $P_d$ such that the following holds:
The length of a shortest $\mathbf{c}$-monotone circuit walk from $\mathbf{s}$ to a $\mathbf{c}$-maximal point of $P$ agrees with the length of a shortest $\mathbf{c}_d$-monotone circuit walk from $\mathbf{s}_d$ to a $\mathbf{c}_d$-maximal point of $P_d$. 
Furthermore, if $P$ has $m$ edges, one can choose $P_d$ to have $m+d-2$ facets.
\end{lemma}
Combining \cref{cor:inapprox} with \cref{thm:anyfixeddim}, we immediately obtain as a consequence the following hardness result in any fixed dimension $d\geq 2$.
\begin{corollary}\label{cor:inapprox-any-d}
For every $\varepsilon>0$ and every $d\in \mathbb{Z}_{\geq 2}$ the following is \NP-hard:
Given as input a $d$-dimensional polytope $P$ with $m$ facets, a starting vertex $\mathbf{s}$, and a vector $\mathbf{c}\in \mathbb{Q}^d$, compute a monotone circuit walk from $\mathbf{s}$ to a $\mathbf{c}$-maximal vertex of $P$ approximating the minimum possible length of such a walk to within a factor of $(m-d)^{1-\varepsilon}$.
\end{corollary}

\paragraph*{Organization.} In \cref{sec:overview} we discuss the necessary definitions and proceed to explain the main ideas of our reductions with the goal of conveying the intuition, without going into all technical details of the proofs. These details are then later supplied in \cref{sec:monotone-circuit-diameter-proof}, where the full proofs of our main technical lemmas are given. 
\section{Overview of the proof}\label{sec:overview}

In this section we will give an overview of the proof of \cref{thm:NP-hard-polygons} and \cref{thm:approx-hardness-monotone-distance}.
We begin by recalling the necessary formal definitions related to circuits.

\begin{definition}[\cite{CircuitPivotRules}]\label{def:circuitmove}
Let $P=\{\mathbf{x}\in \mathbb{R}^d|A\mathbf{x}\le \mathbf{b}\}$ with $A\in \mathbb{R}^{m\times d}, b \in \mathbb{R}^m$ be a polyhedron of dimension $d$.

\begin{enumerate}
    \item A \emph{circuit}\footnote{We remark that in the literature a different (but equivalent) definition of circuits in terms of minimal supports is more prevalent, but for our purposes the definition via submatrices given here is more convenient. The equivalence of our definition and the standard definition as given in~\cite{CircuitPivotRules} can be easily checked, and follows for instance from Lemma~13 in~\cite{BORGWARDT2022}.} of $P$ is a vector $\mathbf{g} \in \mathbb{R}^d\setminus \{\mathbf{0}\}$ for which there exists an index set $I\subseteq [m]$ of size $(d-1)$ such that the $(d-1)\times d$-submatrix $A_I$ of $A$ has rank $d-1$, and such that $A_I\mathbf{g}=0$. 
\item Given a point $\mathbf{x} \in P$, a \emph{circuit move} at $\mathbf{x}$ consists of selecting a circuit $\mathbf{g}$ of $P$ and moving to a new point $\mathbf{x}'=\mathbf{x}+\alpha \mathbf{g}$, where $\alpha > 0$ is \emph{maximal} w.r.t.\ $\mathbf{x}+\alpha \mathbf{g} \in P$.
\item A \emph{circuit walk} of length $k$ is a sequence $(\mathbf{x}_0,\mathbf{x}_1,\ldots,\mathbf{x}_k)$ of points in $P$ such that for every $i=1,\ldots,k$, we have that $\mathbf{x}_i$ is obtained from $\mathbf{x}_{i-1}$ by a circuit move.
\item Given a cost vector $\mathbf{c}\in \mathbb{R}^d$, we say a circuit walk $(\mathbf{x}_0,\mathbf{x}_1,\ldots,\mathbf{x}_k)$ is \emph{$\mathbf{c}$-monotone} if $\mathbf{c}^\top\mathbf{x}_0 < \mathbf{c}^\top\mathbf{x}_1< \ldots<\mathbf{c}^\top\mathbf{x}_k$.
\item Given a point $\mathbf{x} \in P$ and a cost vector $\mathbf{c}\in \mathbb{R}^d$, the \emph{$\mathbf{c}$-monotone circuit distance} $d^P_\mathbf{c}(\mathbf{x})$ from $\mathbf{x}$ is the length of a shortest $\mathbf{c}$-monotone circuit walk that starts in $\mathbf{x}$ and ends in a $\mathbf{c}$-maximal point of $P$.
\end{enumerate}
\end{definition}
As alluded to in the introduction, we will repeatedly use the fact that circuits coincide with edge-directions for polygons.
\begin{observation}[$\ast$]\label{obs:edgesarecircuits}
   Let $A \in \mathbb{R}^{m\times 2}$ and $\mathbf{b}\in \mathbb{R}^m$.
   Let $P = \{\mathbf{x} \in \mathbb{R}^2| A\mathbf{x} \leq \mathbf{b}\}$ be a non-empty polygon.
   If no inequality of $A\mathbf{x}\leq \mathbf{b}$ is redundant, then the circuits of $P$ correspond precisely to the vectors parallel to some edge of $P$.
\end{observation}

As the next ingredient we need to specify the problem we want to reduce to \textsc{Monotone Circuit Distance}.
We will use a reduction from a certain promise variant of subset sum in which we are allowed to use an element more than once.

\begin{mdframed}[innerleftmargin=0.5em, innertopmargin=0.5em, innerrightmargin=0.5em, innerbottommargin=0.5em, userdefinedwidth=0.95\linewidth, align=center]
	{\textsc{Exact Subset Sum with Repetition}}
	\sloppy

	\noindent
	\textbf{Input:} A sequence of distinct non-negative integers $a_1, \dots, a_n, S, k\in \mathbb{Z}_{\geq 0}$, such that $k\leq n$ and 
		$\sum_{i=1}^n r_i a_i = S$ for $r\in \mathbb{Z}^n_{\geq 0}$ implies $\sum_{i=1}^n r_i = k$.

	\noindent
	\textbf{Decision:} Is there a vector $r\in \mathbb{Z}^n_{\geq 0}$ such that
	\[
		\sum_{i=1}^n r_i a_i = S \, \, \, \, \, ?
	\]
\end{mdframed}

Let us motivate the two restrictions of the \textsc{Exact Subset Sum with Repetition} problem a bit more.
The reduction we will present is geometric in nature.
As such, we have to pay attention to the encoding length of the polygons we construct.
It will turn out that the encoding length of the polygon that we construct depends polynomially on $k$, and thus we must enforce $k$ to be polynomial for the instances we consider. 

Additionally, the construction will be designed in such a way that there is a monotone circuit walk of length $2\sum_{i=1}^n r_i$, whenever $\sum_{i=1}^n r_i a_i = S$.
Thus, by adding the promise that $\sum_{i=1}^n r_i = k$ in this case, we will be able to tie the length of the shortest monotone circuit walk to the feasibility of the subset sum instance.
Hardness of the \textsc{Exact Subset Sum with Repetition} problem follows from a standard reduction which we provide in \cref{sec:missing-proofs}.

\begin{theorem}[$\ast$]\label{thm:subset-sum-special-hardness}
	The \textsc{Exact Subset sum with Repetition} problem is \NP-hard.
\end{theorem}

In the following we construct polygons with specific monotone circuit distances.
In later constructions we will then modify these polygons using affine transformations to achieve certain desirable properties.
In order to keep control of monotone circuit distances while performing these transformations, it will be useful to note the following fact. Roughly speaking, it states that affine transformations map monotone circuit walks to monotone circuit walks, and thus also monotone circuit distances (with respect to the mapped objective direction) are preserved under these transformations.

\begin{observation}[$\ast$]\label{obs:transform}
	Let $P = \{\mathbf{x}\in \mathbb{R}^2\colon A\mathbf{x}\leq \mathbf{b}\}$ be a polygon defined by $A\in \mathbb{Q}^{m\times 2}$ and $\mathbf{b}\in \mathbb{Q}^m$.
  Consider an affine transformation defined by an invertible matrix $H\in \mathbb{Q}^{2\times 2}$ and a translation vector $d\in \mathbb{Q}^2$.
	Let $W = (\mathbf{x}_1, \dots, \mathbf{x}_n)$ be a circuit walk in $P$. Then $W':=(H\mathbf{x}_1 + \mathbf{d}, \dots, H\mathbf{x}_n + \mathbf{d})$ is a circuit walk in the transformed polytope $HP + \mathbf{d} = \{\mathbf{x}\in \mathbb{R}^2\colon AH^{-1}\mathbf{x} \leq \mathbf{b} + AH^{-1}\mathbf{d}\}$. Furthermore, if $W$ is $\mathbf{c}$-monotone for some $\mathbf{c}\in \mathbb{R}^2$, then $W'$ is $\mathbf{c}'$-monotone for $\mathbf{c}'\coloneqq (H^\top)^{-1}\mathbf{c}$.
\end{observation}

In our construction, we use as a building block a polygon with large monotone circuit distances.
\textcite{borgwardt2014edges} already showed that there are polygons with large circuit distances (linear in the number of edges). 
However, their proof is of existential nature and does not directly guarantee an efficient construction or a polynomial bound on the encoding length, which are both required for the purposes of our reduction.

Here, we present a new constructive proof for the case of monotone circuit distances which achieves both of these requirements.
\begin{theorem}\label{thm:linear-circuit-distance}
	Given any $\ell\in \mathbb{Z}_{\geq 0}$ one can efficiently determine a matrix $A_\ell\in \mathbb{Z}^{(2\ell+1)\times 2}$ and a vector $\mathbf{b}_\ell\in \mathbb{Z}^{2\ell+1}$, giving a non-redundant description of a polygon $P_\ell = \{\mathbf{x}\in \mathbb{R}^2\colon A_\ell\mathbf{x}\leq \mathbf{b}_\ell\}$ with the following properties:
\begin{enumerate}[label=\textnormal{(\roman*)}]
        \item\label{item:linear-polygon-vertices} The points $\mathbf{u}_\ell \coloneqq (0,1)^\top$ and $\mathbf{w}_\ell \coloneqq (0,-1)^\top$ are vertices of $P_\ell$ spanning an edge of $P_\ell$.
		\item\label{item:linear-polygon-distance} 
        Set $\mathbf{c}_0\coloneqq (1,0)^\top$.
        Then $P_\ell$ has a unique $\mathbf{c}_0$-maximal vertex $\mathbf{t}_\ell$ and the $\mathbf{c}_0$-monotone circuit distance from $\mathbf{u}_\ell$ and $\mathbf{w}_\ell$ to $\mathbf{t}_\ell$ equals $\ell$ each.
		\item\label{item:linear-polygon-encoding} The entries of $A_\ell$ and $\mathbf{b}_\ell$ are each at most $(8\ell+1)^\ell$ in absolute value. In particular, the encoding length of $A_\ell$ and $\mathbf{b}$ is polynomially bounded in $\ell$.
        \item\label{item:linear-polygon-edges} The number of edges of $P_\ell$ is $2\ell + 1$.
	\end{enumerate}
    Additionally, $P_\ell\setminus\{\mathbf{u}_\ell, \mathbf{w}_\ell\}$ lies completely within $\mathbb{R}_{\geq 0}\times (-1, 1)$.
\end{theorem}
Let us briefly motivate the final condition on $P_\ell$.
As mentioned we will use $P_\ell$ as a building block.
In the construction later we want to glue it to another polygon along the edge between $\mathbf{u}_{\ell}$ and $\mathbf{w}_{\ell}$.
The final condition will ensure that this operation yields a convex body.

We postpone the formal proof to \cref{sec:monotone-circuit-diameter-proof}.
Instead we only give an intuition of the construction.
\begin{proof}[Sketch of proof]
We define the polygon $P_\ell$ recursively, using a scaled version of $P_{\ell-1}$ as a building block.
First we define $P_1$ as the triangle $\text{conv}\{(0,1)^\top, (0,-1)^\top, (1,0)^\top\}$.
See \cref{fig:long-monotone-distance-basecase} for a visualization of the construction.
The unique $\mathbf{c}_0$-maximal vertex of $P_1$ is $(1,0)^\top$ and the monotone circuit distance from $\mathbf{u}_1$ and $\mathbf{w}_1$ to $(1,0)^\top$ is one.

	\begin{figure}[ht]
		\begin{subfigure}{0.33\textwidth}
			\begin{center}
				\begin{tikzpicture}
	\begin{scope}[thin, gray]
		\draw[-stealth] (-0.2, 0) -- (2.7,0);
		\node[below, gray, font=\small] at (2.7, 0) {$x$};
		\draw[-stealth] (0,-2.7) -- (0, 2.7);
		\node[left, gray, font=\small] at (0,2.7) {$y$};

		\draw[dashed, gray!60!white,line width=1pt] (2,-2.7) -- (2, 2.7);
		\node[below, gray, font=\small] at (2,-2.7) {$x=1$};
	\end{scope}

	\begin{scope}[radius=2pt, circle, black]
		\coordinate (a1) at (0,2) {};
		\coordinate (b1) at (0,-2) {};
		\coordinate (opt) at (2,0) {};
	\end{scope}

	\begin{scope}[every node/.style={thick,draw=black,fill=black,circle,minimum size=1, inner sep=1}]
		\node at (a1)  {};
		\node at (b1)  {};
		\node at (opt){};
	\end{scope}

	\begin{scope}
		\node[left] at (a1) {$\mathbf{u}_1$};
		\node[left] at (b1) {$\mathbf{w}_1$};
		\node[above right=-3pt] at (opt) {$\mathbf{t}_1$};
	\end{scope}

	\begin{scope}[thick]
		\draw (a1)--(opt);
		\draw (opt) -- (b1);
		\draw (b1) -- (a1);
	\end{scope}
\end{tikzpicture}
 			\end{center}
			\caption{The base case $P_1$.}
			\label{fig:long-monotone-distance-basecase}
		\end{subfigure}
		\begin{subfigure}{0.33\textwidth}
			\begin{center}
				\begin{tikzpicture}
	\begin{scope}[thin, gray]
		\draw[-stealth] (-0.2, 0) -- (2.7,0);
		\node[below, gray, font=\small] at (2.7, 0) {$x$};
		\draw[-stealth] (0,-2.7) -- (0, 2.7);
		\node[left, gray, font=\small] at (0,2.7) {$y$};

		\draw[dashed, gray!60!white,line width=1pt] (2,-2.7) -- (2, 2.7);
		\node[below, font=\small] at (2,-2.7) {$x=1$};
	\end{scope}

	\begin{scope}[radius=2pt, circle, black]
		\coordinate (a2) at (0,2) {};
		\coordinate (a1) at (2,1) {};
		\coordinate (b2) at (0,-2) {};
		\coordinate (b1) at (2,-1) {};
		\coordinate (opt) at (2.25,0) {};
	\end{scope}

	\begin{scope}[every node/.style={thick,draw=black,fill=black,circle,minimum size=1, inner sep=1}]
		\node at (a1)  {};
		\node at (b1)  {};
		\node at (a2)  {};
		\node at (b2)  {};
		\node at (opt){};
	\end{scope}

	\begin{scope}
		\node[above right] at (a1) {$\mathbf{u}_1$};
		\node[left] at (a2) {$\mathbf{u}_2$};
		\node[below right] at (b1) {$\mathbf{w}_1$};
		\node[left] at (b2) {$\mathbf{w}_2$};
		\node[above right=-3pt] at (opt) {$\mathbf{t}_2$};
	\end{scope}

	\begin{scope}[thick]
		\draw (a2) -- (a1) --(opt);
		\draw (opt) -- (b1)-- (b2);
		\draw (b2) -- (a2);
	\end{scope}
\end{tikzpicture}
 			\end{center}
			\caption{Visualization of $P_2$.}
		\end{subfigure}
		\begin{subfigure}{0.33\textwidth}
		\begin{center}
			\begin{tikzpicture}
	\begin{scope}[thin, gray]
		\draw[-stealth] (-0.2, 0) -- (2.7,0);
		\node[below, gray, font=\small] at (2.7, 0) {$x$};
		\draw[-stealth] (0,-2.7) -- (0, 2.7);
		\node[left, gray, font=\small] at (0,2.7) {$y$};

		\draw[dashed, gray!60!white,line width=1pt] (2,-2.7) -- (2, 2.7);
		\node[below, gray, font=\small] at (2,-2.7) {$x=1$};
	\end{scope}

	\begin{scope}[radius=2pt, circle, black]
		\coordinate (a3) at (0,2) {};
		\coordinate (a2) at (2,1) {};
		\coordinate (a1) at (2.25,.3) {};
		\coordinate (b3) at (0,-2) {};
		\coordinate (b2) at (2,-1) {};
		\coordinate (b1) at (2.25,-.3) {};
		\coordinate (opt) at (2.275,0) {};
	\end{scope}

	\begin{scope}[every node/.style={thick,draw=black,fill=black,circle,minimum size=1, inner sep=1}]
		\node at (a1)  {};
		\node at (b1)  {};
		\node at (a2)  {};
		\node at (b2)  {};
		\node at (a3)  {};
		\node at (b3)  {};
		\node at (opt){};
	\end{scope}

	\begin{scope}
		\node[above right=-3pt] at (a1) {$\mathbf{u}_1$};
		\node[above right] at (a2) {$\mathbf{u}_2$};
		\node[left] at (a3) {$\mathbf{u}_3$};
		\node[below right=-3pt] at (b1) {$\mathbf{w}_1$};
		\node[below right] at (b2) {$\mathbf{w}_2$};
		\node[left] at (b3) {$\mathbf{w}_3$};
\end{scope}

	\begin{scope}[thick]
		\draw (a3) -- (a2) -- (a1) --(opt);
		\draw (opt) -- (b1)-- (b2) -- (b3);
		\draw (b3) -- (a3);
	\end{scope}
\end{tikzpicture}
 		\end{center}
		\caption{Visualization of $P_3$.}
		\end{subfigure}
		\caption{Visualization of the polytopes $P_1$, $P_2$, and $P_3$.
$P_{\ell+1}$ is obtained by scaling and shifting $P_\ell$ and adding the vertices $\mathbf{u}_{\ell+1}=(0,1)^\top$ and $\mathbf{w}_{\ell+1}=(0,-1)^\top$.
        For the sake of presentation we denote by $\mathbf{u}_i$ and $\mathbf{w}_i$ the image of $\mathbf{u}_i$ and $\mathbf{w}_i$ under $T_{\ell-1}\circ \dots \circ T_i$ in this figure.
}
		\label{fig:long-monotone-distance-construction}
	\end{figure}
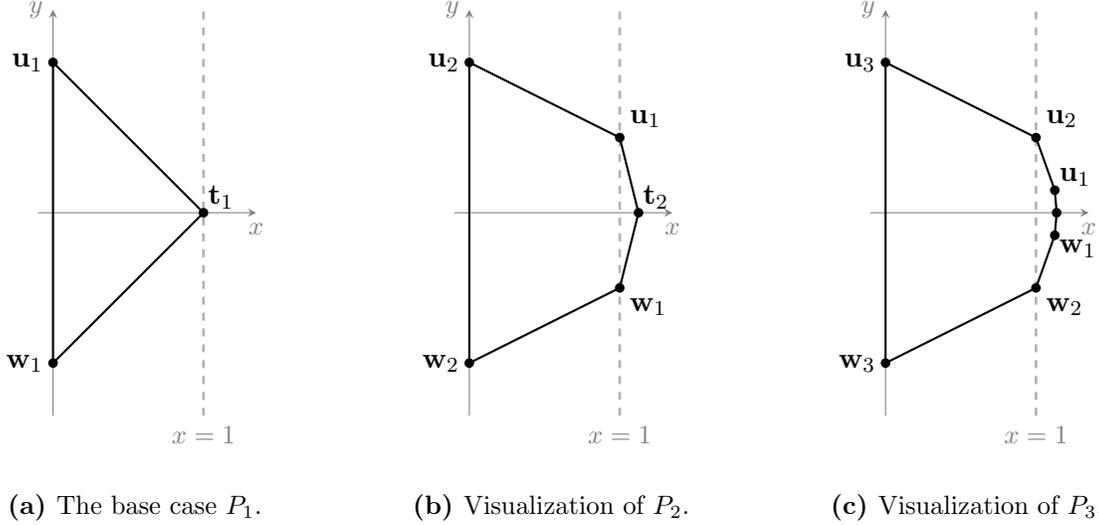
Next we show how to construct $P_{\ell + 1}$ given that we already constructed $P_\ell$.
First we scale $P_\ell$ around the origin by a factor of $\frac{1}{8\ell}$ in the $x$-direction and by a factor of $\frac{1}{2}$ in the $y$-direction.
Next we translate it by $(1,0)^\top$, i.e., by one unit in the $x$-direction.
This moves the vertex $\mathbf{u}_\ell$ to $(1, 0.5)^\top$ and $\mathbf{w}_\ell$ to $(1, -0.5)^\top$.
Let the overall affine transformation be denoted $T_\ell$ and let $T_\ell(P_\ell)$ denote the image of $P_\ell$ under $T_\ell$.

We define $P_{\ell+1}$ as the convex hull of the points $\mathbf{u}_{\ell+1}:=(0,1)^\top$, $\mathbf{w}_{\ell+1}:=(0,-1)^\top$, and the polygon $T_\ell(P_\ell)$.
Note that the non-vertical edge directions of $P_\ell$ have a slope of at least $0.5$ in absolute value.
Thus, all non-vertical edge directions of $T_\ell(P_\ell)$ have a slope of at least $2\ell$ in absolute value.
A visualization of this construction can be seen in \cref{fig:long-monotone-distance-construction}.

We now give the intuition of why $P_{\ell+1}$ satisfies the $\mathbf{c}_0$-monotone circuit distance claimed in \cref{item:linear-polygon-distance}.
The main idea is the following.
Starting from $\mathbf{u}_{\ell+1}$ or $\mathbf{w}_{\ell+1}$ we first show that a $\mathbf{c}_0$-monotone circuit walk of length at most $\ell$ reaching a $\mathbf{c}_0$-maximal vertex has to visit $\mathbf{u}_{\ell}$ or $\mathbf{w}_{\ell}$.
Indeed, the edges of $T_\ell(P_\ell)$ have a slope of at least $2\ell$ in absolute value.
Using that $P_{\ell+1}$ lies in the strip $\mathbb{R}\times [-1,1]$, we observe that a circuit move in $P_{\ell+1}$ using one of these directions changes the $x$-coordinate by at most $\frac{1}{\ell}$.
Hence, using $\ell$ of these moves we cannot reach a point outside of the edge between $\mathbf{u}_{\ell+1}$ and $\mathbf{u}_\ell$ and the edge between $\mathbf{w}_{\ell+1}$ and $\mathbf{w}_\ell$.
Call these edges $e$ and $f$, respectively.
Thus, at some point we need to use the direction of $e$ or $f$ for a circuit move starting at a point on $e$ or $f$.
Now the $\mathbf{c}_0$-increasing orientation of these directions lead to $\mathbf{u}_{\ell}$ or $\mathbf{w}_\ell$, proving that we have to visit one of the two vertices.

To finish the proof observe that starting from points on the boundary of $P_{\ell+1}$ coming from $T_\ell(P_\ell)$ we can only take circuits parallel to edges of $T_\ell(P_\ell)$.
The other edges do not provide directions that are both $\mathbf{c}_0$-monotone and feasible at any point under consideration.
Thus, applying $T_\ell^{-1}$ to any $\mathbf{c}_0$-monotone circuit walk starting at $\mathbf{u}_{\ell}$ or $\mathbf{w}_\ell$ gives rise to a $\mathbf{c}_0$-monotone circuit walk in $P_\ell$.
This allows us to finish as we thus have $d^{P_{\ell+1}}_{\mathbf{c}_0}(\mathbf{u}_\ell) = d^{P_{\ell+1}}_{\mathbf{c}_0}(\mathbf{w}_\ell) = \ell$.
Here we used $d^{P_\ell}_\mathbf{c}(\mathbf{u}_\ell) = d^{P_\ell}_\mathbf{c}(\mathbf{w}_\ell) = \ell$.
\end{proof}

Our main results, \cref{thm:NP-hard-polygons} and \cref{thm:approx-hardness-monotone-distance}, follow from the following theorem.
The theorem shows that we can encode feasibility of an instance of the \textsc{Exact Subset Sum with Repetition} problem using the monotone circuit distance of a certain polygon.
In the construction we can choose a constant $C$ that encodes the gap in monotone circuit distance we can achieve between a feasible and an infeasible instance.
The number of edges of the polygon depends on $C$ as well.
Maximizing $C$ while maintaining a polynomial encoding length and constructibility of the associated polygon then yields the $m^{1-\varepsilon}$-inapproximability stated in \cref{thm:approx-hardness-monotone-distance}.

\begin{theorem}\label{thm:monotone-circuit-diameter}
	Let $(a_1, \dots, a_n, S, k)$ be an instance of the \textsc{Subset sum with Repetition} problem.
	Additionally, assume we are given $C \in \mathbb{Z}_{\geq 0}$.
	There is a polygon $P$ with a vertex $\mathbf{s}$ and a cost vector $\mathbf{c}$ such that the following holds:
	\begin{enumerate}[label=\textnormal{(\roman*)}]
		\item\label{item:main-short-walk} If the \textsc{Subset sum with Repetition} instance is feasible, then $d^P_{\mathbf{c}}(\mathbf{s}) \leq 2k$.
		\item\label{item:main-long-walk} If the \textsc{Subset sum with Repetition} instance is infeasible, then $d^P_{\mathbf{c}}(\mathbf{s}) > Ck$.
		\item\label{item:main-construct} We can construct $P$, $\mathbf{s}$, and $\mathbf{c}$ in time polynomial in $n$, $\log S$, $k$, and $C$.
		\item\label{item:main-encoding} The encoding length of $P$, $\mathbf{s}$, and $\mathbf{c}$ is polynomial in $n$, $\log S$, $k$, and $C$.
		\item\label{item:main-edges} The number of edges of $P$ is bounded by $2Ck + 2n$.
	\end{enumerate}
\end{theorem}

Before proving \cref{thm:monotone-circuit-diameter}, let us directly demonstrate how it implies \cref{thm:approx-hardness-monotone-distance}.

\begin{proof}[Proof of \cref{thm:approx-hardness-monotone-distance} assuming \cref{thm:monotone-circuit-diameter}]
	Let $(a_1, \dots, a_n, S, k)$ be an instance of the \textsc{Exact Subset sum with Repetition} problem.
    Note that by definition of the problem, we then have $k\le n$. 
	We can construct a polygon $P$ with a vertex $\mathbf{s}$ and a cost vector $\mathbf{c}$ as in \cref{thm:monotone-circuit-diameter}, where we set $C = \left\lceil\max\left\{ 8^\frac{1}{\varepsilon}k^{\frac{1-\varepsilon}{\varepsilon}}, 8n^{1-\varepsilon} \right\}\right\rceil$.
	Note that for any fixed $\varepsilon > 0$ we have that $C$ is polynomially bounded in terms of $n$. 
    Hence, by \cref{thm:monotone-circuit-diameter}, $P$ can be constructed in time polynomial in $n$, $\log S$, and also its encoding length is polynomially bounded in $n$ and $\log S$. In particular this implies that construction time and encoding length of $P$ are polynomially bounded in terms of the description length of the input $(a_1,\ldots,a_n,S,k)$ to the \textsc{Exact Subset Sum with Repetition} problem.
    By \cref{thm:monotone-circuit-diameter} \ref{item:main-edges}, the number $m$ of edges of $P$ is at most $2Ck + 2n$.

    Note that we chose $C$ in a way that we have
	\[
		\left(\frac{k^{1-\varepsilon}}{C^{\varepsilon}} + \frac{n^{1 - \varepsilon}}{C}\right) \leq \frac{1}{4}\enspace . 
	\]
	Indeed, the first term in the maximum defining $C$ ensures that the first summand is at most $\frac{1}{8}$ and the second term ensures that the second summand is at most $\frac{1}{8}$.
    Thus, we in particular have
	\[
        2\left((Ck)^{1-\varepsilon} + n^{1 - \varepsilon}\right) (2k) \leq Ck\enspace .
	\]
    Now note that the left hand side is larger than $(2Ck + 2n)^{1-\varepsilon} (2k)$, as $(x + y)^{1-\varepsilon}\leq  x^{1 - \varepsilon} + y^{1- \varepsilon}$ for any $x,y\geq 0$.
    In particular, we have $m^{1-\varepsilon} (2k) \leq Ck$.

    This inequality, combined with \cref{thm:monotone-circuit-diameter} \ref{item:main-short-walk} and \ref{item:main-long-walk} now implies that $(P,\mathbf{s},\mathbf{c}, 2k)$ is a polynomial-size instance of \textsc{Monotone Circuit Distance}, which satisfies either $d^P_\mathbf{c}(\mathbf{s})\le 2k$ (if the \textsc{Subset Sum with Repetition} instance $(a_1,\ldots,a_n,S)$ is feasible) or $d^P_\mathbf{c}(\mathbf{s})>m^{1-\varepsilon}(2k)$ (if $(a_1,\ldots,a_n,S)$ is infeasible).
    Hence, this provides a polynomial reduction of \textsc{Exact Subset Sum with Repetition} to the special case of the \textsc{Monotone Circuit Distance} problem with the additional constraint on instances specified in the statement of \cref{thm:approx-hardness-monotone-distance}.
    Since \textsc{Exact Subset Sum with Repetition} is \NP-hard by \cref{thm:subset-sum-special-hardness}, this concludes the proof of \cref{thm:approx-hardness-monotone-distance}.
\end{proof}

The remainder of this article is dedicated to the proof of \cref{thm:monotone-circuit-diameter}.
In \cref{sec:reduction-overview} we give an overview of the reduction and the intuition behind it.
This is followed by a detailed proof of \cref{thm:linear-circuit-distance} and \cref{thm:monotone-circuit-diameter} in \cref{sec:monotone-circuit-diameter-proof}.

\subsection{Overview of the reduction}\label{sec:reduction-overview}

Before getting into the technical details, let us describe the intuition behind the proof of \cref{thm:monotone-circuit-diameter}.
Consider an instance $(a_1, \dots, a_n, S, k)$ of the \textsc{Exact Subset sum with Repetition} problem.

We start with the rectangle $[0,1]\times [0,S+\varepsilon]$, slightly taller than the target number.
The precise value of $\varepsilon$ will be determined later.
We replace the upper left and the lower right corner of the rectangle with two polygonal chains.
The rough idea is that the lower right replacement gives rise to a circuit direction of slope $a_i$ for every $i\in [n]$.
The replacement of the upper left corner will give a vertex $\mathbf{t}$ at height $S$.
Let $\mathbf{c}\in \mathbb{R}^2$ be a vector such that $\mathbf{t}$ is the unique $\mathbf{c}$-maximum.
The idea of the construction is to ensure that, starting at $\mathbf{s}=(0,0)^\top$, $\mathbf{t}$ will essentially only be reachable by a short $\mathbf{c}$-monotone circuit walk, if we use the circuit directions with slope $a_i$ to reach height precisely $S$.
This will in turn give rise to a solution to the \textsc{Exact Subset sum with Repetition} problem.
See \cref{fig:monotone-construction-overview} for a visualization of the instance we will construct.

Let $r\in \mathbb{Z}_{\geq 0}^n$ be a solution to $(a_1,\dots, a_n, S,k)$, i.e., $\sum_{i=1}^{n} r_i=k$ and $\sum_{i=1}^{\ell} r_i a_i = S$.
Consider the circuit walk starting at $\mathbf{s}$ in which we alternatingly take a circuit direction with slope $a_i$ and then circuit direction $(-1,0)^\top$.
Here, we take the direction corresponding to $a_i$ precisely $r_i$ times.
We perform the replacement of the upper left corner in such a way that it does not modify points with a $y$ coordinate below $S-0.5$.
Then the circuit moves with slope $a_i$ all start at the left edge and change the $y$-coordinate by precisely $a_i$.
In total, we change the $y$-coordinate by $\sum_{i=1}^{n} r_i a_i = S$, and thus reach the point $(1, S)^\top$ after $2k-1$ steps.
In particular, we can reach $\mathbf{t}$ in $2k$ circuit moves, by using the circuit direction $(-1,0)^\top$ once more.
The main difficulty is to construct the instance, in particular the replacement of the upper left corner, in such a way that every monotone circuit walk of length at most $Ck$ gives rise to a solution of the subset sum problem.

\begin{figure}[h!]
	\begin{subfigure}[t]{0.5\textwidth}
		\begin{center}
			\begin{tikzpicture}
	\newcommand{\s}{3}
	\newcommand{\e}{0.3}
	\newcommand{\x}{2}

	\begin{scope}[thin, gray]
	\draw[-stealth] (-0.2, 0) -- (\x + 0.7,0);
	\node[below, gray, font=\small] at (\x + 0.7, 0) {$x$};
	\draw[-stealth] (0,-.2) -- (0, \s+1.2);
	\node[left, gray, font=\small] at (0,\s+1.2) {$y$};

	\draw[dashed, gray!60!white,line width=1pt] (\x,-.2) -- (\x, \s+.7);
	\node[below, gray, font=\small] at (\x,-.2) {$1$};

	\draw[dashed, gray!60!white,line width=1pt] (-.1,\s) -- (\x+.6, \s);
	\node[right, gray, font=\small] at (\x+.6,\s) {$S$};

	\draw[dashed, gray!60!white,line width=1pt] (-.1,\s+\e) -- (\x+.6, \s+\e);
	\node[right, gray, font=\small] at (\x+.6,\s+\e) {$S+\varepsilon$};

	\draw[dashed, gray!60!white,line width=1pt] (-.1,\s-\e) -- (\x+.6, \s-\e);
	\node[right, gray, font=\small] at (\x+.6,\s-\e) {$S-\varepsilon$};

	\draw[dashed, gray!60!white,line width=1pt] (-.1,.2) -- (\x+.6, .2);
	\node[right, gray, font=\small] at (\x+.6,.2) {};
	\end{scope}

	\begin{scope}
		\coordinate (s) at (0,0);
		\coordinate (u) at (\x,\s + \e);

		\coordinate (v0) at (0   , \s - \e);
		\coordinate (t0) at (0.1 , \s - \e + .15);
		\coordinate (t)  at (0.25, \s      );
		\coordinate (t1) at (0.5, \s + 0.2);
		\coordinate (v1) at (0.8 , \s + \e);

		\coordinate (a0) at (\x      , 0.2);
		\coordinate (a1) at (\x - 0.1, 0.1);
		\coordinate (a2) at (\x - 0.3, 0.05);
		\coordinate (a3) at (\x - 0.7, 0);
	\end{scope}

	\begin{scope}[thick]
		\draw (s) -- (v0) -- (t0) -- (t) -- (t1) -- (v1) -- (u) -- (a0) -- (a1) -- (a2) -- (a3) -- (s);
	\end{scope}

	\node[xshift = -1.5cm, yshift=-.2cm] (ulabel) at (v0) {$T(\mathbf{u}_{Ck})$};
	\draw[-stealth,shorten >=1pt] (ulabel) -- (v0);
	\node[above left = 0.5cm] (tlabel) at (t) {$\mathbf{t}$};
	\draw[-stealth,shorten >=1pt] (tlabel) -- (t);
	\node[xshift = .2cm, yshift = 1cm] (wlabel) at (v1) {$T(\mathbf{w}_{Ck})$};
	\draw[-stealth,shorten >=1pt] (wlabel) -- (v1);
	\node[below left] at (s) {$\mathbf{s}$};
\end{tikzpicture}
 		\end{center}
		\caption{The construction used in the reduction.
		Starting from a rectangle we replace the upper left and lower right corners.}
	\end{subfigure}
	\begin{subfigure}[t]{0.5\textwidth}
		\begin{center}
			\begin{tikzpicture}
	\newcommand{\s}{3}
	\newcommand{\e}{0.3}
	\newcommand{\x}{2}

	\begin{scope}[thin, gray]
		\draw[-stealth] (-0.2, 0) -- (\x + 0.7,0);
		\node[below, gray, font=\small] at (\x + 0.7, 0) {$x$};
		\draw[-stealth] (0,-.2) -- (0, \s+1.2);
		\node[left, gray, font=\small] at (0,\s+1.2) {$y$};

		\draw[dashed, gray!60!white,line width=1pt] (\x,-.2) -- (\x, \s+.7);
		\node[below, gray, font=\small] at (\x,-.2) {$1$};

		\draw[dashed, gray!60!white,line width=1pt] (\x,\s) -- (\x+.6, \s);
		\node[right, gray, font=\small] at (\x+.6,\s) {$S$};
	\end{scope}

	\begin{scope}
		\coordinate (s) at (0,0);
		\coordinate (u) at (\x,\s + \e);

		\coordinate (v0) at (0   , \s - \e);
		\coordinate (t0) at (0.1 , \s - \e + .15);
		\coordinate (t)  at (0.25, \s      );
		\coordinate (t1) at (0.5, \s + 0.2);
		\coordinate (v1) at (0.8 , \s + \e);

		\coordinate (a0) at (\x      , 0.2);
		\coordinate (a1) at (\x - 0.1, 0.1);
		\coordinate (a2) at (\x - 0.3, 0.05);
		\coordinate (a3) at (\x - 0.7, 0);
	\end{scope}

	\begin{scope}[red, thick]
		\draw (s) -- (\x, 0.8);
		\draw (\x, 0.8) -- (0,0.8);
		\draw (0, 0.8) -- (\x, 1.6);
		\draw (\x, 1.6) -- (0, 1.6);
		\draw (0, 1.6) -- (\x, \s);
		\draw (\x, \s) -- (t);
	\end{scope}

	\newcommand{\braceshift}{.1}
	\begin{scope}[thick, shorten <=1pt, shorten >= 1pt]
		\draw[decorate,decoration={brace,amplitude=3pt,mirror,raise=2pt}] (\x+\braceshift, 0) --node[right=3pt]{$a_1$} (\x+\braceshift, 0.8);
		\draw[decorate,decoration={brace,amplitude=3pt,mirror,raise=2pt}] (\x+\braceshift, 0.8) --node[right=3pt]{$a_1$} (\x+\braceshift, 1.6);
		\draw[decorate,decoration={brace,amplitude=3pt,mirror,raise=2pt}] (\x+\braceshift, 1.6) --node[right=3pt]{$a_2$} (\x+\braceshift, \s);
	\end{scope}

	\begin{scope}[thick]
		\draw (s) -- (v0) -- (t0) -- (t) -- (t1) -- (v1) -- (u) -- (a0) -- (a1) -- (a2) -- (a3) -- (s);
	\end{scope}

	\node[above left = 0.5cm] (tlabel) at (t) {$\mathbf{t}$};
	\draw[-stealth,shorten >=1pt] (tlabel) -- (t);
	\node[below left] at (s) {{$\mathbf{s}$}};
\end{tikzpicture}
 		\end{center}
		\caption{Visualization of a circuit walk of length $2k$ which can be obtained from a solution to the Subset Sum instance.
		In this case, $r_1 = 2, r_2 = 1$, and $r_1 a_1 + r_2 a_2 = S$.}
	\end{subfigure}
	\caption{Visualization of the idea behind the reduction.
		Starting from a rectangle we replace the upper left corner by an affine transformation $T$ of the polygon $P_{Ck}$ constructed in \cref{thm:linear-circuit-distance}.
		Additionally, we replace the lower right corner by edges with slopes corresponding to the elements of the subset sum instance.
		We want to find a circuit walk from the lower left corner to the middle vertex $\mathbf{t}$ of the upper left vertices.
		Note that in order to increase readability the scale of polytope is not the same as in the actual construction.}
	\label{fig:monotone-construction-overview}
\end{figure}
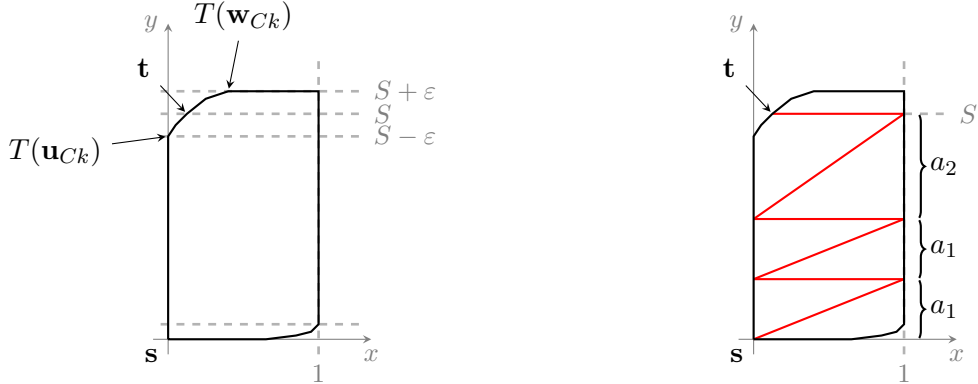

We use the polygon $P_{Ck}$ constructed in \cref{thm:linear-circuit-distance}, where we set $\ell = Ck$.
First, consider the polygonal chain obtained from the boundary of $P_{Ck}$ by removing the edge between $u_{Ck}$ and $w_{Ck}$.
We replace the upper left corner of the rectangle by the image of this polygonal chain under an affine transformation $T$.
Here, $\mathbf{u}_{Ck}$ is mapped to a point on the $x=0$ line and $\mathbf{w}_{Ck}$ is mapped to a point on the $y=S+\varepsilon$ line.
The remainder of the polygonal chain is mapped above the line through the images of $\mathbf{u}_{Ck}$ and $\mathbf{w}_{Ck}$.
Write $T$ as $T(\mathbf{x}) = H\mathbf{x} + \mathbf{b}$, for an invertible matrix $H\in \mathbb{R}^{2\times 2}$ and a vector $\mathbf{b}\in \mathbb{R}^2$.
Set the cost vector $\mathbf{c}$ to $(H^{-1})^\top \mathbf{c}_0$.
This reflects the change in cost identified in \cref{obs:transform}, allowing us to later transform $\mathbf{c}$-monotone circuit walks in $P$ to $\mathbf{c}_0$-monotone circuit walks in $P_{Ck}$.
The value of $\varepsilon$ is determined by the transformation, with the details following later.
In order for the reduction to work, we have to ensure that the transformation $T$ has the following properties:
First, the unique $\mathbf{c}_0$-maximal vertex $\mathbf{t}_{Ck}$ of $P_{Ck}$ shall be mapped to a unique $\mathbf{c}$-maximal vertex $\mathbf{t}$ of the new polygon with a $y$-value of $S$.
Second, we will need that the edges of $T(P_{Ck})$ have sufficiently small slope.
To be precise, their slope shall be smaller than $\frac{1}{2Ck}$.
Third, all vertices on $T(P_{Ck})$ shall be close to the point $(0, S)^\top$, which will be made precise later.

The first condition ensures that a solution to the \textsc{Exact Subset sum with Repetition} instance gives rise to a $\mathbf{c}$-monotone circuit path from $\mathbf{s}$ to $\mathbf{t}$ of length $Ck$ through the interior of the polytope.
On the other hand, if the subset sum instance is infeasible, then the properties of $P_{Ck}$ will guarantee that every circuit walk from $\mathbf{s}$ to $\mathbf{t}$ has length at least $Ck+1$.

The second condition ensures that taking $Ck$ circuit directions corresponding to the edges of $P_{Ck}$ can change the $y$-coordinate by at most $\frac{1}{2}$ in total.
In particular we cannot ``cheat'' in height by taking directions that do not correspond to elements of the subset sum instance.
The third condition is necessary for the following reason.
Assume we could find a short circuit walk that reaches a point $\mathbf{p}$ lying on the ``left'' edge, slightly below $T(\mathbf{u}_{Ck})$.
We cannot exclude that $\mathbf{t}$ is reachable in a single circuit move from $\mathbf{p}$.
This would circumvent the circuit distance guarantee we get from $P_{Ck}$.
So, in order to prohibit ``short-cutting'', we will choose the scaling of $P_{Ck}$ carefully.
In the end, every point $\mathbf{p}$ that allows us to take a shortcut will have a $y$-coordinate close to $S$.
Then any short circuit walk reaching $\mathbf{p}$ still gives rise to a solution of the subset sum instance.
The precise scaling will be determined later.
Roughly speaking we want to ensure that every $\mathbf{c}$-monotone circuit walk of length at most $Ck$ starting from a point on the upper edge with large $x$-coordinate that reaches the image of $P_{Ck}$ visits $T(\mathbf{u}_{Ck})$ or $T(\mathbf{w}_{Ck})$.

Finally, we have to introduce circuit directions of slopes $a_1, \dots, a_n$ by replacing the lower right corner of the rectangle with a certain polygonal chain.
Without loss of generality, assume $a_1 < a_2 < \dots < a_n$.
We replace the lower right corner of the rectangle by $n+1$ vertices, $\mathbf{v}_0, \dots, \mathbf{v}_{n}$.
The slope of the edge between $\mathbf{v}_{i-1}$ and $\mathbf{v}_{i}$ will be precisely $a_i$.
Additionally, we choose the scaling small enough, such that the replacement does not impact monotone walks.
To be precise, we will make sure that $\mathbf{c}^\top\mathbf{v}_i \leq \mathbf{c}^\top\mathbf{s}$ for all $i\in \{0,\dots,n\}$, so we can never visit any of the newly constructed edges in a monotone circuit walk starting at $\mathbf{s}$.

Using this construction we will show that if the subset sum instance is feasible, then there exists a $\mathbf{c}$-monotone circuit walk from $\mathbf{s}$ to $\mathbf{t}$ of length at most $2k$.
Conversely, if the subset sum instance is infeasible, then every $\mathbf{c}$-monotone circuit walk from $\mathbf{s}$ to $\mathbf{t}$ has length at least $Ck + 1$.
Let us briefly repeat the argument for the first of these two statements more explicitly.
Consider an $r\in \mathbb{Z}_{\geq 0}^n$ that is a solution to the \textsc{Exact Subset sum with Repetition} instance, i.e., $\sum_{i=1}^{n} r_i =k$ and $\sum_{i=1}^{n} r_i a_i = S$.
Then, we can construct a circuit walk from $\mathbf{s}$ to $\mathbf{t}$ of length $2k$.
To be precise, let $b_1, \dots, b_k$ be a sequence of elements containing each $a_i$ exactly $r_i$ times.
Then the following is a $\mathbf{c}$-monotone circuit walk from $\mathbf{s}$ to $\mathbf{t}$: \[(0,0)^\top \rightarrow (1, b_1)^\top \rightarrow (0, b_1)^\top \rightarrow (1, b_1 + b_2)^\top \rightarrow (0, b_1 + b_2)^\top \rightarrow \dots \rightarrow \left(1, \sum_{i=1}^k b_i\right)^\top \rightarrow \mathbf{t}\enspace .\]
In the last step we used that the $y$-coordinate of $\mathbf{t}$ is precisely $S = \sum_{i=1}^{k} b_i$.

We will now give a rough intuition for the second statement.
Consider a $\mathbf{c}$-monotone circuit walk $W$ from $\mathbf{s}$ to $\mathbf{t}$ of length at most $Ck$.
The circuit directions corresponding to the $a_i$'s cannot be used in the upper left corner.
Thus, by the distance lower bound of $P_{Ck}$, $W$ cannot visit $\mathbf{u}_{Ck}$ or $\mathbf{w}_{Ck}$.
Hence, $W$ has to ``shortcut'' into the upper left corner.
Our choice of the affine transformation applied to $P_{Ck}$ will enforce that $W$ must contain a point with a $y$-coordinate in the interval $(S-\frac{1}{2}, S+\varepsilon)$, before visiting any point on the upper edge.
Then $W$ allows us to find a feasible solution to the \textsc{Exact Subset sum with Repetition} instance.
This is due to the slopes of the edges of $T(P_{Ck})$ being small, i.e., smaller than $\frac{1}{2Ck}$.
So, in total, circuit moves corresponding to these edges can change the $y$-coordinate by at most $\frac{Ck}{2Ck} = \frac{1}{2}$.
Now consider the circuit moves of $W$ that have a slope of $a_i$ for some $i\in [n]$.
In order to reach a point with $y$-coordinate in $(S-\frac{1}{2}, S+\varepsilon)$, the change in $y$-coordinate due to these circuit moves has to be precisely $S$.
Let $r_i$ be the number of times $W$ uses a circuit move with slope $a_i$.
As we do not visit the ``upper'' edge, a move with slope $a_i$ changes the $y$-coordinate by precisely $a_i$.
Rephrasing the above, we have $\sum_{i=1}^{n} r_i a_i = S$. 
Hence, $r$ gives a solution to the \textsc{Exact Subset Sum with Repetition} instance.
In particular, we have $\sum_{i=1}^{n} r_i = k$ and $d^P_{\mathbf{c}}(\mathbf{s}) \leq 2k$.

\section{Formal construction and proofs}\label{sec:monotone-circuit-diameter-proof}

In this section we formalize the intuition given before and prove \cref{thm:linear-circuit-distance} and \cref{thm:monotone-circuit-diameter}.
We begin with a formal proof of \cref{thm:linear-circuit-distance}.

\begin{proof}[Proof of \cref{thm:linear-circuit-distance}]
    For a visualization of the proof we refer back to \cref{fig:long-monotone-distance-construction}.
    We define the polygon $P_\ell$ recursively.
    Additionally, we maintain the following invariant:
	All edge directions (and by Observation~\ref{obs:edgesarecircuits} all circuits) of $P_\ell$ are either parallel to $(0,1)^\top$ or their slope is at least $0.5$ in absolute value.

	First we define $P_1$ as the triangle $\text{conv}\{(0,1)^\top, (0,-1)^\top, (1,0)^\top\}$.
    Recall that $\mathbf{c}_0=(1,0)^\top$.
	See \cref{fig:long-monotone-distance-basecase} for a visualization of the construction.
	Then the unique $\mathbf{c}_0$-maximal vertex of $P_1$ is $(1,0)^\top$ and the monotone circuit distance from $\mathbf{u}_1$ and $\mathbf{w}_1$ to $(1,0)^\top$ is one. One furthermore easily checks that the remaining items and the invariant hold with this definition of $P_1$, encoded by 
    $$A_1\coloneqq \begin{pmatrix}
-1 & 0\\
1 & 1 \\
1 & -1
\end{pmatrix}, \mathbf{b}_1\coloneqq \begin{pmatrix}
0\\
1 \\
1
\end{pmatrix}.$$

	Next we will show how to construct $P_{\ell + 1}$ given that we already constructed $P_\ell$.
    We assume that $P_\ell$ satisfies the invariant.
	Define $P_{\ell+1}$ the following way:
	Scale $P_\ell$ around the origin by a factor of $\frac{1}{8\ell}$ in the $x$-direction and by a factor of $\frac{1}{2}$ in the $y$-direction.
	Next translate it by $(1,0)^\top$, i.e., by one unit in the $x$-direction.
	This moves the vertex $\mathbf{u}_\ell$ to $(1, 0.5)^\top$ and $\mathbf{w}_\ell$ to $(1, -0.5)^\top$.
	Let the overall affine transformation be denoted $T_\ell$ and let $T_\ell(P_\ell)$ denote the image of $P_\ell$ under $T_\ell$.

	We define $P_{\ell+1}$ as the convex hull of the points $\mathbf{u}_{\ell+1}:=(0,1)^\top$, $\mathbf{w}_{\ell+1}:=(0,-1)^\top$, and the polygon $T_\ell(P_\ell)$. It follows directly by this definition that $\mathbf{u}_{\ell+1}, \mathbf{w}_{\ell+1}$ are vertices of $P_{\ell+1}$ that span an edge, verifying that \cref{item:linear-polygon-vertices} is satisfied. Also note that $P_{\ell+1}$ has a unique $\mathbf{c}_0$-maximal vertex, namely the image of the $\mathbf{c}_0$-maximal vertex of $P_\ell$.
    
	By our invariant, all non-vertical edge directions of $P_\ell$ have a slope of at least $0.5$ in absolute value.
	Thus, all non-vertical edge directions of $T_\ell(P_\ell)$ have a slope of at least $2\ell>1$ in absolute value.
    Additionally, the slopes from $\mathbf{u}_{\ell+1}=(0,1)^\top$ to $T_\ell(\mathbf{u}_\ell)=(1,0.5)^\top$ and from $\mathbf{u}_{\ell+1}=(0,-1)^\top$ to $T_\ell(\mathbf{w}_{\ell})=(1,-0.5)^\top$ are $-0.5$ and $0.5$, respectively.
	This implies that $P_{\ell+1}$ contains all vertices of $T_\ell(P_\ell)$ as vertices, and that $\mathbf{u}_\ell$ and $T_\ell(\mathbf{u}_\ell)$ as well as $\mathbf{w}_\ell$ and $T_\ell(\mathbf{w}_\ell)$ are connected by edges of $P_{\ell+1}$. It follows that our invariant remains satisfied for $P_{\ell+1}$.

    Finally, let $P_\ell=\{\mathbf{x}\in \mathbb{R}^2|A_\ell \mathbf{x}\le \mathbf{b}_\ell\}$ be the inequality description of $P_\ell$. With out loss of generality, let the first row of $A_\ell$ be $(-1,0)$ and the first entry of $\mathbf{b}_\ell$ be $0$, corresponding to the edge-defining inequality $x\ge 0$ of $P_\ell$.
    Then we have
    \[P_{\ell+1}=\{\mathbf{x}\in \mathbb{R}^2|A_{\ell+1}\mathbf{x}\le \mathbf{b}_{\ell+1}\},\]
    where 
    \[A_{\ell+1}\in \mathbb{Z}^{(2\ell+3)\times 2}\]
    is obtained from $A_\ell\in \mathbb{Z}^{(2\ell+1)\times 2}$ by multiplying all entries of the first column but the first by $8\ell$, multiplying all entries of the second column but the first by $2$, and then adding two new last rows $(1,2)$ and $(1,-2)$ at the bottom.
    Similarly, $b_{\ell+1}\in \mathbb{Z}^{2\ell+3}$ is obtained from $b_\ell\in \mathbb{Z}^{2\ell+1}$ by keeping the first entry as is (i.e., $0$), then adding for every $i\in [2\ell+1]$ the first entry of $A_{\ell+1}$ in the $i$-th row to the $i$-th entry, and finally inserting two new last entries, the first of which is $1$ and the second of which is $-1$.
    It is not hard to check that $A_{\ell+1}$ and $\mathbf{b}_{\ell+1}$ indeed describe $P_{\ell+1}$.

    One checks that the maximum absolute value of an entry of $A_{\ell+1}, b_{\ell+1}$ defined in this way is by at most a factor $8\ell+1$ larger than the maximum absolute value among entries in $A_\ell, \mathbf{b}_\ell$, which was assumed to be at most $(8\ell+1)^\ell$ by \cref{item:linear-polygon-encoding}.
    Thus, the maximum absolute value among entries in $A_{\ell+1}, \mathbf{b}_{\ell+1}$ is at most $(8\ell+1)^{\ell+1}<(8(\ell+1)+1)^{\ell+1}$, showing that \cref{item:linear-polygon-encoding} is also satisfied by this linear description of $P_{\ell+1}$. 
    
	A visualization of this construction can be seen in \cref{fig:long-monotone-distance-construction}.
    
	Let us now prove that $P_{\ell+1}$ satisfies \cref{item:linear-polygon-distance}.
    To show this, it suffices to prove $d^{P_{\ell+1}}_{\mathbf{c}_0}(\mathbf{u}_{\ell+1}) = \ell + 1$ as the proof for $d^{P_{\ell+1}}_{\mathbf{c}_0}(\mathbf{w}_{\ell+1}) = \ell + 1$ can be obtained the same way, exploiting symmetry along the $y=0$ axis.
	We first show $d^{P_{\ell+1}}_{\mathbf{c}_0}(\mathbf{u}_{\ell+1}) \ge \ell + 1$.
    Towards a contradiction, assume there was a $\mathbf{c}_0$-monotone circuit walk $W$ in $P_{\ell+1}$ from $\mathbf{u}_{\ell+1}$ to a $\mathbf{c}_0$-maximal vertex of length at most $\ell$ (as $\mathbf{c}_0=(1,0)^\top$, this circuit walk strictly increases the $x$-coordinate at every step).
    Note that by \cref{obs:edgesarecircuits} the circuits of $P_{\ell+1}$ are given by the following: (1) The vectors parallel to non-vertical edge-directions of $T_\ell(P_\ell)$ and $(2)$ vectors parallel to $(0,1)^\top$ or to $(1,\pm 0.5)^\top$. 
    Recall further that the non-vertical edges of $T_\ell(P_\ell)$ have a slope of at least $2\ell$ in absolute value.
	As $P_{\ell+1}$ lies in the strip $\mathbb{R}\times [-1,1]$, any circuit move in $P_{\ell+1}$ using one of these directions changes the $x$-coordinate by at most $\frac{1}{\ell}$.
    This implies that $W$ must use one of the edge directions $(1,\pm 0.5)$ at least once before reaching a point with $x$-coordinate bigger than $1$: Otherwise, the maximum $x$-coordinate reached by $W$ would be at most $\ell\cdot \frac{1}{\ell}=1$, contradicting that $W$ ends in the unique $\mathbf{c}_0$-maximal vertex who has an $x$-coordinate bigger than one.
    Note that since $W$ is monotone in the $x$-coordinate, it does not visit any point on the ``left'' edge between $\mathbf{u}_{\ell+1},\mathbf{w}_{\ell+1}$ except for $\mathbf{u}_{\ell+1}$.
    Thus, $W$ must contain a circuit move in direction $(1, \pm 0.5)^\top$ starting from a point on one of the edges spanned between $\mathbf{u}_{\ell+1}$ and $(1,0.5)^\top$ or between $\mathbf{w}_{\ell+1}$ and $(1,-0.5)^\top$.
	In particular $W$ must visit $T_\ell(\mathbf{u}_\ell) = (1, 0.5)^\top$ or $T_\ell(\mathbf{w}_\ell) = (1, -0.5)^\top$, as any circuit move with the above property ends in one of $T_\ell(\mathbf{u}_\ell)$,  $T_\ell(\mathbf{w}_\ell)$.
    
	Consider now the suffix $W'$ of $W$ starting from $T_\ell(\mathbf{u}_\ell)$ or $T_\ell(\mathbf{w}_\ell)$.
    Note that the length of $W'$ is less than~$\ell$.
	By $\mathbf{c}_0$-monotonicity $W'$ visits only points on $T_\ell(P_\ell)$.
	The slope of the edges of $T_\ell(P_\ell)$ is smaller in absolute value than $1$, so a circuit step in direction $(1, \pm 0.5)^\top$ is not feasible at any point on the part of the boundary of $P_{\ell+1}$ that is contained in $T_\ell(P_\ell)$.
    Thus, $W'$ only uses circuit directions parallel to non-vertical edges of $T_\ell(P_\ell)$.
    By \cref{obs:transform}, scaling $W'$ then gives rise to a circuit walk in $P_\ell$.
    The latter starts at $\mathbf{u}_\ell$ or $\mathbf{w}_\ell$ and reaches the unique $\mathbf{c}_0$-maximal vertex of $P_\ell$ with less than $\ell$ steps.
    This contradicts that the $\mathbf{c}_0$-monotone circuit distance from $\mathbf{u}_\ell$ and $\mathbf{w}_\ell$ in $P_\ell$ equals $\ell$ (by our assumptions on $P_\ell$).
    Hence, our assumption that a circuit walk $W$ from $\mathbf{u}_{\ell+1}$ to a $\mathbf{c}_0$-maximal vertex of length at most $\ell$ exists was wrong, proving that $d^{P_{\ell+1}}_{\mathbf{c}_0}(\mathbf{u}_{\ell+1})\ge \ell+1$.
    In the other direction, observe that one can reach $T_\ell(\mathbf{u}_\ell)$ from $\mathbf{u}_{\ell+1}$ with a single circuit move, so $d^{P_{\ell+1}}_{\mathbf{c}_0}(\mathbf{u}_{\ell+1} )\leq d^{P_{\ell+1}}_{\mathbf{c}_0}(T_\ell(\mathbf{u}_{\ell})) + 1 \leq \ell +1 $.
    Here we used \cref{obs:transform} and the properties of $P_{\ell}$ for the last inequality.
    Thus, we indeed have $d^{P_{\ell+1}}_{\mathbf{c}_0}(\mathbf{u}_{\ell+1})=\ell+1$, as desired, finishing the proof of \cref{item:linear-polygon-distance}.
    
	Finally, observe that the construction consists of a sequence of $\ell$ linear transformations of polynomial size in $\ell$, giving \cref{item:linear-polygon-encoding}.
	Hence, it can be done in time polynomial in $\ell$ and the encoding length stays polynomial as well.
    Furthermore, $P_1$ has three edges and we add two edges when constructing $P_{\ell+1}$ from $P_\ell$.
    Thus, the number of edges is also as claimed in \cref{item:linear-polygon-edges}.
    By construction we also have $P_{\ell + 1}\setminus\{\mathbf{u}_{\ell + 1}, \mathbf{w}_{\ell + 1}\}\subseteq \mathbb{R}_{\geq 0}\times (-1,1)$.
\end{proof}

In the remainder of the article we will prove \cref{thm:monotone-circuit-diameter}.
In order to do so, we first need to introduce some notation.
For a point $\mathbf{v} = (a,b)^\top\in \mathbb{R}^2$ we denote by $\mathbf{v}^x := a, \mathbf{v}^y := b$ its $x$- and $y$-coordinate, respectively.

We explicitly construct a polygon $P$ with a vertex $\mathbf{s}$ and a cost vector $\mathbf{c}$ satisfying the conditions of \cref{thm:monotone-circuit-diameter}.
To start the construction of the polytope $P$, we begin with the replacement of the upper left corner.
In order to use the replacement in the overall construction, we summarize below the key properties we are going to use.

\begin{lemma}\label{lemma:replacement-hard-instance}
    Let $(a_1, \dots, a_n, S, k)$ be an instance of the \textsc{Exact Subset sum with Repetition} problem with $0\leq a_1 < a_2 < \dots < a_n$ and let $C\in \mathbb{Z}_{\geq 0}$ be given.
    We can efficiently determine an affine transformation $T\colon \mathbb{R}^2 \rightarrow \mathbb{R}^2$ with the following properties.
    \begin{enumerate}[label=\textnormal{(\roman*)}]
        \item\label{item:hard-slopes} For every edge of $T(P_{Ck})$ its slope $s$ satisfies $0 < s < \frac{1}{2Ck}$.
        \item\label{item:hard-coordinates} 
        Let $\mathbf{t}_{Ck}$ denote the unique $\mathbf{c}_0$-maximal vertex of $P_{Ck}$.
        Then $T(\mathbf{t}_{Ck})^y = S$ and $T(\mathbf{u}_{Ck})^x = 0$.
        Furthermore, for every $\mathbf{p}\in T(P_{Ck})$ we have 
        \[
            0 \leq \mathbf{p}^x < \left(\frac{s_1}{a_n}\right)^{\left\lceil\frac{Ck}{2}\right\rceil+1}
            \qquad \text{and} \qquad
            S - \frac{1}{2}\left(\frac{s_1}{a_n}\right)^{\left\lceil\frac{Ck}{2}\right\rceil+1} < \mathbf{p}^y < S + \frac{1}{2}\left(\frac{s_1}{a_n}\right)^{\left\lceil\frac{Ck}{2}\right\rceil+1}\enspace.
        \]
        \item\label{item:hard-vertex} Every vertex $\mathbf{v}$ of $T(P_{Ck})$ is $\mathbf{c}_\mathbf{v}$-maximal for a vector $\mathbf{c}_\mathbf{v}\in \mathbb{R}^2$ with $\mathbf{c}_\mathbf{v}^x < 0$ and $\mathbf{c}_\mathbf{v}^y > 0$.
        \item\label{item:hard-encoding} The encoding length of $T(P_{Ck})$ is polynomial in $a_1, \dots, a_n$, $k$, $\log S$, and $C$.
    \end{enumerate}
\end{lemma}

\begin{proof}
Let $P_{Ck}$ be the polygon constructed in \cref{thm:linear-circuit-distance} with vertices $\mathbf{u}_{Ck}$ and $\mathbf{w}_{Ck}$.
Furthermore, let $\mathbf{t}_{Ck}$ be the unique $\mathbf{c}_0$-maximal vertex of $P_{Ck}$.
In order to make the construction more digestible, we will proceed in several steps.
The first step roughly orients $P_{Ck}$ and ensures that all slopes are positive and bounded.
The second step guarantees that the slopes satisfy \cref{item:hard-slopes}.
The third step ensures that all vertices are close to each other.
Finally, we align the polytope with the left edge and the image of $\mathbf{t}_{Ck}$ with the $y=S$ line, to establish \cref{item:hard-coordinates}.
In the end we will check that \cref{item:hard-vertex} and \cref{item:hard-encoding} hold as well.
A visualization of the affine transformations can be seen in \cref{fig:linear-transformation}.

\begin{figure}[h]
	\centering
    \begin{subfigure}[t]{0.32\textwidth}
       \begin{center}
       \begin{tikzpicture}[scale=.8]
	\begin{scope}[thin, gray]
		\draw[-stealth] (-0.2, 0) -- (2.7,0);
		\node[below, gray, font=\small] at (2.7, 0) {$x$};
		\draw[-stealth] (0,-2.7) -- (0, 2.7);
		\node[left, gray, font=\small] at (0,2.7) {$y$};

		\draw[dashed, gray!60!white,line width=1pt] (2,-2.7) -- (2, 2.7);
		\node[above right, gray, font=\small] at (2,-2.7) {$x=1$};
	\end{scope}

	\begin{scope}[radius=2pt, circle, black]
		\coordinate (a2) at (0,2) {};
		\coordinate (a1) at (2,1) {};
		\coordinate (b2) at (0,-2) {};
		\coordinate (b1) at (2,-1) {};
		\coordinate (opt) at (2.25,0) {};
	\end{scope}

	\begin{scope}[every node/.style={thick,draw=black,fill=black,circle,minimum size=1, inner sep=1}]
		\node at (a1)  {};
		\node at (b1)  {};
		\node at (a2)  {};
		\node at (b2)  {};
		\node at (opt){};
	\end{scope}

	\begin{scope}[thick]
		\draw (a2) -- (a1) --(opt);
		\draw (opt) -- (b1)-- (b2);
		\draw (b2) -- (a2);
	\end{scope}
\end{tikzpicture}        \end{center}
       \caption{We start with the polygon constructed in \cref{thm:linear-circuit-distance}.}
    \end{subfigure}
    \begin{subfigure}[t]{0.32\textwidth}
       \begin{center}
       \begin{tikzpicture}[scale=.8]
	\begin{scope}[thin, gray]
		\draw[-stealth] (-.2, 0) -- (2.7,0);
		\node[below, gray, font=\small] at (2.7, 0) {$x$};
		\draw[-stealth] (0,-2.7) -- (0, 2.7);
		\node[left, gray, font=\small] at (0,2.7) {$y$};
\node[yshift = 3ex, xshift = 1ex, gray, font=\small] at (1,0) {$\begin{psmallmatrix}
		    0.5 \\ 0
		\end{psmallmatrix}$};
		\draw[dashed, gray!60!white,line width=1pt] (0,-2) -- (1, 0);
		\draw[dashed, gray!60!white,line width=1pt] (0, 2) -- (1, 0);
	\end{scope}

\begin{scope}[xscale=0.2]
	\begin{scope}[radius=2pt, circle, black]
		\coordinate (a2) at (0,2) {};
		\coordinate (a1) at (2,1) {};
		\coordinate (b2) at (0,-2) {};
		\coordinate (b1) at (2,-1) {};
		\coordinate (opt) at (2.25,0) {};
	\end{scope}
\end{scope}

	\begin{scope}[every node/.style={thick,draw=black,fill=black,circle,minimum size=1, inner sep=1}]
		\node at (a1)  {};
		\node at (b1)  {};
		\node at (a2)  {};
		\node at (b2)  {};
		\node at (opt){};
	\end{scope}

	\begin{scope}[thick]
		\draw (a2) -- (a1) --(opt);
		\draw (opt) -- (b1)-- (b2);
		\draw (b2) -- (a2);
	\end{scope}
\end{tikzpicture}
        \end{center}
       \caption{We scale in $x$-direction such that the polygon lives inside the given triangle.}
    \end{subfigure}
    \begin{subfigure}[t]{0.32\textwidth}
       \begin{center}
       
\begin{tikzpicture}[scale=.8]
	\begin{scope}[thin, gray]
		\draw[-stealth] (-2.7, 0) -- (2.7,0);
		\node[below, gray, font=\small] at (2.7, 0) {$x$};
		\draw[-stealth] (0,-2.7) -- (0, 2.7);
		\node[left, gray, font=\small] at (0,2.7) {$y$};
		\draw[dashed, gray!60!white,line width=1pt] (-2.2,-2.2) -- (2.2, 2.2);
		\node[below, gray, font=\small] at (-2.2,-2.2) {$x=y$};
	\end{scope}

\begin{scope}[shift={(0,0)}, rotate=135]
\draw[dashed, gray!60!white,line width=1pt] (0,-2) -- (1, 0);
\draw[dashed, gray!60!white,line width=1pt] (0, 2) -- (1, 0);
\node[above left, gray] at (1,0) {
$\begin{psmallmatrix}    
-\sqrt{2}/4 \\    \sqrt{2}/4
\end{psmallmatrix}$

}; 
\begin{scope}[xscale=0.2]
	\begin{scope}[radius=2pt, circle, black]
		\coordinate (a2) at (0,2) {};
		\coordinate (a1) at (2,1) {};
		\coordinate (b2) at (0,-2) {};
		\coordinate (b1) at (2,-1) {};
		\coordinate (opt) at (2.25,0) {};
	\end{scope}
\end{scope}
\end{scope}

	\begin{scope}[every node/.style={thick,draw=black,fill=black,circle,minimum size=1, inner sep=1}]
		\node at (a1)  {};
		\node at (b1)  {};
		\node at (a2)  {};
		\node at (b2)  {};
		\node at (opt){};
	\end{scope}

	\begin{scope}[thick]
		\draw (a2) -- (a1) --(opt);
		\draw (opt) -- (b1)-- (b2);
		\draw (b2) -- (a2);
	\end{scope}
\end{tikzpicture}        \end{center}
       \caption{We rotate, so that all slopes are in $[1/3, 3]$.}
    \end{subfigure}
    \begin{subfigure}[t]{0.32\textwidth}
        \begin{center}
       
\begin{tikzpicture}[scale=.8]
	\begin{scope}[thin, gray]
		\draw[-stealth] (-1.7, 0) -- (2.7,0);
		\node[below, gray, font=\small] at (2.7, 0) {$x$};
		\draw[-stealth] (0,-.7) -- (0, 4.2);
		\node[left, gray, font=\small] at (0,4.2) {$y$};
	\end{scope}
		\node[left, font=\small] at (0,5.2) {\phantom{$y$}};

\begin{scope}[yscale=.2]
\begin{scope}[shift={(0,0)}, rotate=135]
\begin{scope}[xscale=0.2]
	\begin{scope}[radius=2pt, circle, black]
		\coordinate (a2) at (0,2) {};
		\coordinate (a1) at (2,1) {};
		\coordinate (b2) at (0,-2) {};
		\coordinate (b1) at (2,-1) {};
		\coordinate (opt) at (2.25,0) {};
	\end{scope}
\end{scope}
\end{scope}

\end{scope}

\newcommand{\x}{-0.7071}
\newcommand{\y}{-0.1414}
\node[left, gray, font=\small] at (\x,\y+1.2) {\phantom{y}};
	\begin{scope}[every node/.style={thick,draw=black,fill=black,circle,minimum size=1, inner sep=1}]
		\node at (a1)  {};
		\node at (b1)  {};
		\node at (a2)  {};
		\node at (b2)  {};
		\node at (opt){};
	\end{scope}

	\begin{scope}[thick]
		\draw (a2) -- (a1) --(opt);
		\draw (opt) -- (b1)-- (b2);
		\draw (b2) -- (a2);
	\end{scope}
\end{tikzpicture}        \end{center}
       \caption{Next we scale down in $y$-direction so that all slopes are small.}
    \end{subfigure}
    \begin{subfigure}[t]{0.32\textwidth}
        \begin{center}
       
\begin{tikzpicture}[scale=.8]
	\begin{scope}[thin, gray]
		\draw[-stealth] (-1.7, 0) -- (2.7,0);
		\node[below, gray, font=\small] at (2.7, 0) {$x$};
		\draw[-stealth] (0,-.7) -- (0, 4.2);
		\node[left, gray, font=\small] at (0,4.2) {$y$};
	\end{scope}

		\node[left, font=\small] at (0,5.2) {\phantom{$y$}};

\begin{scope}[scale=.25]
\begin{scope}[yscale=.2]
\begin{scope}[shift={(0,0)}, rotate=135]
\begin{scope}[xscale=0.2]
	\begin{scope}[radius=2pt, circle, black]
		\coordinate (a2) at (0,2) {};
		\coordinate (a1) at (2,1) {};
		\coordinate (b2) at (0,-2) {};
		\coordinate (b1) at (2,-1) {};
		\coordinate (opt) at (2.25,0) {};
	\end{scope}
\end{scope}
\end{scope}
\end{scope}

\end{scope}

\newcommand{\x}{-0.7071}
\newcommand{\y}{-0.1414}

	\begin{scope}[every node/.style={thick,draw=black,fill=black,circle,minimum size=1, inner sep=1}]
		\node at (a1)  {};
		\node at (b1)  {};
		\node at (a2)  {};
		\node at (b2)  {};
		\node at (opt){};
	\end{scope}

    \begin{scope}[scale = 0.5]
        
    \draw[decorate,decoration={brace,amplitude=2pt,raise=1ex, mirror}] (-\x, -\y) --node[midway, yshift=3ex, xshift=1ex]{$a$} (\x, -\y);
    \draw[decorate,decoration={brace,amplitude=2pt,raise=1ex, mirror}] (-\x, \y) --node[midway, xshift= 3ex]{$b$} (-\x, -\y);
    \end{scope}

	\begin{scope}[thick]
		\draw (a2) -- (a1) --(opt);
		\draw (opt) -- (b1)-- (b2);
		\draw (b2) -- (a2);
	\end{scope}
\end{tikzpicture}        \end{center}
       \caption{We scale the whole polygon so that its width $a$ and height $b$ are both bounded by $\left(\frac{s_1}{a_n}\right)^{\left\lceil\frac{Ck}{2}\right\rceil + 1}$.}
       \label{subfig:transformation_5}
    \end{subfigure}
    \begin{subfigure}[t]{0.32\textwidth}
        \begin{center}
       \begin{tikzpicture}[scale=.8]

\begin{scope}[shift={(-0.3535, 0.0707)}]
\begin{scope}[scale=.25]
\begin{scope}[yscale=.2]
\begin{scope}[shift={(0,0)}, rotate=135]
\begin{scope}[xscale=0.2]
	\begin{scope}\coordinate (a2) at (0,2) {};
		\coordinate (a1) at (2,1) {};
		\coordinate (b2) at (0,-2) {};
		\coordinate (b1) at (2,-1) {};
		\coordinate (opt) at (2.25,0) {};
	\end{scope}
\end{scope}
\end{scope}
\end{scope}

\end{scope}
\end{scope}

\newcommand{\x}{-0.7071}
\newcommand{\y}{-0.1414}

		\draw[white] (\x,\y-3.7) -- (\x, \y +1.2);
	\begin{scope}[thin, gray]
\draw[-stealth] (\x -.2, \y - 3) -- (\x + 2.7,\y-3);
		\node[below, gray, font=\small] at (\x + 2.7, \y - 3) {$x$};
		\draw[-stealth] (\x,\y-3.7) -- (\x, \y +1.2);
		\node[left, gray, font=\small] at (\x,\y+1.2) {$y$};
	\end{scope}

    \coordinate (upper-right) at (1, -\y);
    \coordinate (s) at (\x, \y - 3);
    \coordinate (lower-right) at (1, \y-3); 
    \coordinate (l3) at (1, \y-2.6); 
    \coordinate (l2) at (.9, \y-2.9); 
    \coordinate (l1) at (.7, \y-3); 

	\begin{scope}[every node/.style={thick,draw=black,fill=black,circle,minimum size=1, inner sep=1}]
		\node at (a1)  {};
		\node at (b1)  {};
		\node at (a2)  {};
		\node at (b2)  {};
		\node at (opt){};
	\end{scope}

	\begin{scope}[thick]
		\draw (a2) -- (a1) --(opt);
		\draw (opt) -- (b1)-- (b2);

		\draw (b2) -- (a2);
        \begin{scope}[gray]
            
        \draw (b2) -- (upper-right);
        \draw (upper-right) -- (l3);
        \draw (l3) -- (l2);
        \draw (l2) -- (l1);
        \draw (l1) -- (s);
        \draw (s) -- (a2);
        \end{scope}
	\end{scope}
\end{tikzpicture}        \end{center}
       \caption{Finally, we translate the polygon.
       In the following we will use it as part of the sketched polygon.}
    \end{subfigure}
    \caption{Visualization of the affine transformation described in \cref{lemma:replacement-hard-instance}.
    To increase visibility, the scaling in \cref{subfig:transformation_5} is not as in the actual construction.
    }
    \label{fig:linear-transformation}
\end{figure}
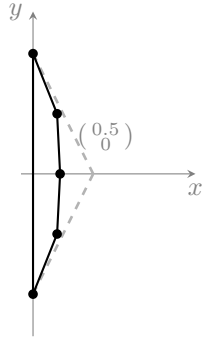
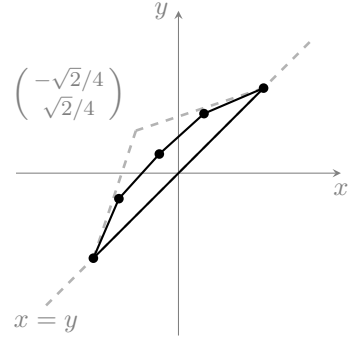
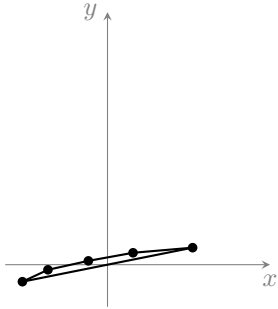
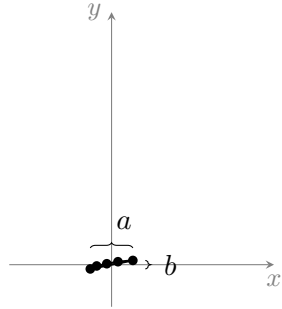
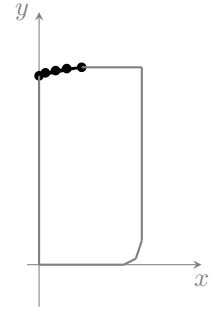
We will now describe the affine transformation we are using.
Recall that $P_{Ck}\setminus \{\mathbf{u}_{Ck}, \mathbf{w}_{Ck}\}\subseteq \mathbb{R}_{\geq 0}\times (-1, 1)$.
We first use a scaling $S$ along the $x$-axis, i.e., a function of the form $S((x,y)^\top) = (\alpha x, y)^\top$ for an $\alpha\in \mathbb{R}_{\geq 0}$.
We chose $\alpha$ such that $S(P_{Ck}) \subseteq \text{conv}((0,1)^\top, (0,-1)^\top, (0.5,0)^\top)$.
This can be done by setting $\alpha \coloneqq \frac{1}{2}\max\{\frac{\mathbf{p}^x}{1-|\mathbf{p}^y|}| \mathbf{p}\in \text{vertices}(P_{Ck})\setminus\{\mathbf{u}_{Ck},\mathbf{w}_{Ck}\}\}$, ensuring polynomial encoding length of the image.
Next consider the rotation $R$ around the origin with $R(u_{Ck}) =\left(-{\sqrt{2}}/{2}, -{\sqrt{2}}/{2}\right)^\top$.
Note that $R(\mathbf{w}_{Ck})=\left({\sqrt{2}}/{2}, {\sqrt{2}}/{2}\right)^\top$ and $R((0.5, 0)^\top) = (-\sqrt{2}/4, \sqrt{2}/4)^\top$.
Set $T_1 = R \circ S$.
Note that we have
\[T_1(P_{Ck})\subseteq \text{conv}\left(
\left(-\frac{\sqrt{2}}{2}, -\frac{\sqrt{2}}{2}\right)^\top,
\left(\frac{\sqrt{2}}{2}, \frac{\sqrt{2}}{2}\right)^\top,
\left(-\frac{\sqrt{2}}{4}, \frac{\sqrt{2}}{4}\right)^\top
\right)\enspace .\]

Furthermore, any slope of $T_1(P_{Ck})$ is upper-bounded by the slope between 
$\left(-{\sqrt{2}}/{2}, -\sqrt{2}/{2}\right)^\top$
and
$\left(-\sqrt{2}/{4}, \sqrt{2}/{4}\right)^\top$
and lower-bounded by the slope between
$\left(-{\sqrt{2}}/{4}, {\sqrt{2}}/{4}\right)^\top$
and
$\left({\sqrt{2}}/{2}, {2}/{2}\right)^\top$.
The former edge has a slope of $3$ and the latter edge has a slope of $1/3$.
Hence the slopes of $T_1(P_{Ck})$ all lie in the interval $[1/3, 3]$.

The next affine transformation $T_2$ scales along the $y$-axis, i.e., we set $T_2((x,y)^\top) = (x, \beta y)^\top$ for $\beta\in \mathbb{R}_{\geq 0}$.
We choose $\beta$ such that after the transformation all edges have a slope in the interval $\left(0,\frac{1}{2Ck}\right)$, e.g., by setting $\beta \coloneqq \frac{1}{6Ck}$.
Here we use that all slopes of $T_1(P_{Ck})$ are bounded by $3$.
Note that the encoding length of $\beta$ and hence the encoding length of of $T_2(T_1(P_{Ck}))$ is polynomial in the input.

In the next step, we scale the polytope as a whole, i.e., we determine a transformation $T_3$ with $T_3((x,y)^\top) = (\gamma x, \gamma y)^\top$ for some $\gamma \in \mathbb{R}_{\geq 0}$.
Let $s_1$ be the smallest slope of $T_2(T_1(P_{Ck}))$ and let $a_n$ be the largest element of the subset sum instance.
We choose $\gamma$ such that the length of $T_3(T_2(T_1(P_{Ck})))$ in $x$ and $y$ direction is bounded by $\left(\frac{s_1}{a_n}\right)^{\left\lceil\frac{Ck}{2}\right\rceil+1}$ each.
This can be done by setting $\gamma = \frac{1}{2}\left(\frac{s_1}{a_n}\right)^{\left\lceil\frac{Ck}{2}\right\rceil+1}$, as $T_2(T_1(P_{Ck}))\subseteq [-1,1]^2$.
Since $a_n \leq S$ and since $s_1$ has polynomial encoding length, the encoding length stays polynomial in $n$, $\log S$, $k$, and $C$.

Finally, translate the polygon (uniquely) in such a way that $T_2(T_1(\mathbf{u}_{Ck}))$ gets mapped to a point on the ($x=0$)--axis and $T_2(T_1(\mathbf{t}))$ to a point on the ($y=S$)--axis.
This keeps the encoding length of the polygon polynomially bounded.
Let $T$ denote the combined affine transformation.

We check that $T$ satisfies all the properties of the lemma.
First, by definition of $T_2$ we know that all slopes of $T_2(T_1(P_{Ck}))$ are between $0$ and $\frac{1}{2Ck}$.
As $T$ is obtained from $T_2\circ T_1$ by composing with a scaling and a translation, the slopes of $T(P_{Ck})$ agree with the slopes of $T_2(T_1(P_{Ck}))$, proving~\cref{item:hard-slopes}.
\cref{item:hard-coordinates} is satisfied by construction of $T_3$ and the definition of the final translation.
We already argued that the encoding length remains polynomial, showing \cref{item:hard-encoding}.
Finally, consider \cref{item:hard-vertex}.
Consider the edges of $T(P_{Ck})$ but the edge spanned by $T(\mathbf{u}_{Ck})$ and $T(\mathbf{w}_{Ck})$.
Note that by \cref{thm:linear-circuit-distance} \ref{item:linear-polygon-edges} there are $2Ck$ such edges.
Let $s_1<\dots< s_{2Ck}$ denote the slopes of these edges.
By \cref{item:hard-slopes} we have $0<s_1$ and $s_{2Ck} < \frac{1}{2Ck}$
Furthermore, in order to simplify notation for the next argument, let the vertices of $T(P_{Ck})$ be denoted $\mathbf{p}_{2Ck}, \dots, \mathbf{p}_{0}$, sorted by increasing $x$-coordinate.
We choose the indices in reverse so that for $i\in [2Ck-1]$ the vertex $\mathbf{p}_i$ is incident to the edge with slope $s_{i+1}$ to its left and the edge with slope $s_{i}$ to its right.
Additionally, $\mathbf{p}_{2Ck} = T(\mathbf{u}_{Ck})$ and $\mathbf{p}_{0} = T(\mathbf{w}_{Ck})$.
We now set
\[
  \mathbf{c}_{\mathbf{p}_i} = \begin{cases}
    (-\frac{s_1}{2}, 1)^\top & \text{if } i=0,\\
    (-2{s_{2Ck}}, 1)^\top& \text{if } i=2Ck,\\
    \left(-\frac{s_i + s_{i+1}}{2},1\right)^\top & \text{else.}
  \end{cases}
\]
The vector $\mathbf{p}_{j}-\mathbf{p}_{j-1}$ is for every $j\in [2Ck]$ obtained from $(-1, -s_j)^\top$ by multiplying with a positive scalar.
Since we have
\[
    \mathbf{c}_{\mathbf{p}_i}^\top (1, s_j)^\top = \frac{s_i + s_{i+1}}{2} - s_j \begin{cases}
        > 0 & \text{for } j \leq i, \\
        < 0 & \text{for } j > i,
    \end{cases}
\]
for every $i\in [2Ck-1]$, it follows that
\[
   \mathbf{c}_{\mathbf{p}_i}^\top \mathbf{p}_0 < \mathbf{c}_{\mathbf{p}_i}^\top \mathbf{p}_1 < \dots < \mathbf{c}_{\mathbf{p}_i}^\top \mathbf{p}_{i-1} < \mathbf{c}_{\mathbf{p}_i}^\top \mathbf{p}_i > \mathbf{c}_{\mathbf{p}_i}^\top \mathbf{p}_{i+1} > \dots > \mathbf{c}_{\mathbf{p}_i}^\top \mathbf{p}_{Ck}\enspace .
\]
for all $i\in [2Ck-1]$, as desired.
A similar computation for $\mathbf{c}_{\mathbf{p}_0}$ and $\mathbf{c}_{\mathbf{p}_{2Ck}}$ shows that $\mathbf{c}_{\mathbf{p}_0}^\top \mathbf{p}_0>\mathbf{c}_{\mathbf{p}_0}^T\mathbf{p}_i$ for all $i\in [2Ck]$ and $\mathbf{c}_{\mathbf{p}_{2Ck}}^\top \mathbf{p}_{2Ck}>\mathbf{c}_{\mathbf{p}_{2Ck}}^\top \mathbf{p}_i$ for all $i\in \{0,\ldots,2Ck-1\}$, thus establishing the statement of \cref{item:hard-vertex} in all cases.
This concludes the proof of the lemma.
\end{proof}

Next, we define the replacement of the lower right corner of the rectangle.
We ensure that no $\mathbf{c}$-monotone circuit walk visits the vertices in the lower right corner for a $\mathbf{c}$ making $T(\mathbf{t})$ maximal.
To do so, we ensure that their $\mathbf{c}$-objective value is worse than the $\mathbf{c}$-objective value of $\mathbf{s}$.
Here we only use that we have $\mathbf{c}^x < 0$ and $\mathbf{c}^y>0$.
Again, we summarize all needed properties in a lemma.

\begin{lemma}\label{lemma:replacement-correct-slopes}
    Let $(a_1, \dots, a_n, S, k)$ be an instance of the \textsc{Exact Subset sum with Repetition} problem, with $0\leq a_1 < a_2 < \dots < a_n$.
    Additionally, let $\mathbf{c}\in\mathbb{R}^2$ with $\mathbf{c}^x < 0$ and $\mathbf{c}^y > 0$ be given.
    Then, given these inputs, we can efficiently determine a set of points $\mathcal{V} =\{\mathbf{v}_0, \dots, \mathbf{v}_n\}\subseteq \mathbb{R}^2$ with the following properties:
    \begin{enumerate}[label=\textnormal{(\roman*)}]
        \item\label{item:A-coordinates} We have $\mathbf{v}_0^y = 0$ and $\mathbf{v}_n^x = 1$.
        Furthermore, for every $\mathbf{v}\in \mathcal{V}$ it holds that 
        \[
            0 < \mathbf{v}^x \leq 1 \qquad \text{and} \qquad 0 \leq \mathbf{v}^y < 1\enspace.
        \]
        \item\label{item:A-slopes} The slope of the line segment from $\mathbf{v}_{i-1}$ to $\mathbf{v}_i$ is $a_i$, for every $i\in \{0,\dots, n\}$. 
        \item\label{item:A-cost} We have $\mathbf{c}^\top \mathbf{v} \leq 0$, for every $\mathbf{v} \in \mathcal{V}$.
        \item\label{item:A-vertex} For every $i\in \{0,\dots,n\}$ there exists a $\mathbf{c}_i\in \mathbb{R}^2$ such that $\mathbf{c}_1^x > 0$, $\mathbf{c}_i^y < 0$, and $\mathbf{v}_i$ is the unique element in  $\mathrm{argmax}\{\mathbf{c}_i^\top \mathbf{v}|\mathbf{v}\in \mathcal{V}\}.$
        \item\label{item:A-encoding} The encoding length of $\mathbf{v}_0, \dots, \mathbf{v}_n$ is polynomial in $a_1, \dots, a_n$, and the encoding length of $\mathbf{c}$.
    \end{enumerate}
\end{lemma}

\begin{proof}
Set $\beta = \min\{-\frac{\mathbf{c}^x}{\mathbf{c}^y}, 1\}$.
By assumption on $\mathbf{c}$ we have $\beta \in (0,1]$, and the encoding length of $\beta$ is polynomial in the encoding length of $\mathbf{c}$. 
Denote by $f_i \coloneqq \sum_{j=1}^i a_j$ for $i\in \{0, 1,\dots, n\}$ the partial sums of $a_1,\ldots,a_n$.
Define $\mathbf{v}_i \coloneqq \left(1 - \frac{(n-i) \beta}{f_n}, \frac{f_i \beta}{f_n}\right)^\top$ for $i\in \{0, 1, \dots, n\}$.
Note that $f_n \geq \frac{n(n-1)}{2}$, as $0 \leq a_1 < a_2 < \dots < a_n$.
Using this, one easily checks that \ref{item:A-coordinates} is satisfied.
Furthermore, note that with this definition, for every $i$ the point $\mathbf{v}_i$ has polynomial encoding length in terms of the encoding length of $(a_1,\ldots,a_n)$ and $\mathbf{c}$, certifying that \ref{item:A-encoding} holds. 

Using the above coordinates we can compute the slope of the line segment spanned by $\mathbf{v}_{i-1}$ and $\mathbf{v}_{i}$ as $\frac{(f_i - f_{i-1})\beta}{f_n}\frac{f_n}{\beta} = f_i - f_{i-1} = a_i$, establishing \ref{item:A-slopes}.
Note that the definition of $\beta$ implies $\beta\le 1$ and $\mathbf{c}^x\le -\beta \mathbf{c}^y$.
Hence, we obtain for every $i\in \{0,\ldots,n\}$:

\begin{align*}
\mathbf{c}^\top \mathbf{v}_i&=\mathbf{c}^x\cdot\left(1-\frac{(n-i)\beta}{f_n}\right)+\mathbf{c}^y\cdot \frac{f_i\beta}{f_n}\\
&\le -\beta \mathbf{c}^y\cdot\left(1-\frac{(n-i)\beta}{f_n}\right)+\mathbf{c}^y\cdot \frac{f_i\beta}{f_n}\\
&=\beta \mathbf{c}^y\left(\frac{(n-i)\beta+f_i}{f_n}-1\right)\le 0 \enspace ,
\end{align*}
where the last inequality follows since $(n-i)\beta+f_i\le n-i+f_i\le f_n$.
Here we used that $a_1,\ldots,a_n\ge 1$, implies $f_n-f_i\ge n-i$.
This shows that \ref{item:A-cost} is satisfied.

Next, let us define the directions $\mathbf{c}_0,\ldots,\mathbf{c}_n$ for \ref{item:A-vertex}.
The definition and the proof will be almost identical to the proof of \cref{item:hard-vertex} of \cref{lemma:replacement-hard-instance}.
We set 
\[
\mathbf{c}_i \coloneqq 
\begin{cases}
    \left(\frac{a_1}{2}, -1\right)^\top & \text{if } i = 0,\\
    (2a_n, -1)^\top & \text{if } i = n,\\ 
    \left(\frac{a_{i} + a_{i+1}}{2}, -1\right)^\top & \text{else.}
\end{cases}
\]
We claim that $\{\mathbf{v}_i\}=\mathrm{argmax}\{\mathbf{c}_i^\top \mathbf{v}|\mathbf{v}\in \mathcal{V}\}$.
By \cref{item:A-slopes}, for all $j\in [n]$ the vector $\mathbf{v}_j-\mathbf{v}_{j-1}$ is obtained from $(1, a_j)^\top$ by multiplying with a positive scalar.
Since we have 
\[
    \mathbf{c}_i^\top (1, a_j)^\top = \frac{a_i + a_{i+1}}{2} - a_j \begin{cases}
        > 0 & \text{for } j \leq i, \\
        < 0 & \text{for } j > i,
    \end{cases}
\]
for every $i\in [n-1]$, it follows that
\[
   \mathbf{c}_i^\top \mathbf{v}_0 < \mathbf{c}_i^\top \mathbf{v}_1 < \dots < \mathbf{c}_i^\top \mathbf{v}_{i-1} < \mathbf{c}_i^\top \mathbf{v}_i > \mathbf{c}_i^\top \mathbf{v}_{i+1} > \dots > \mathbf{c}_i^\top \mathbf{v}_n\enspace .
\]
for all $i\in [n-1]$, as desired.
A similar computation for $\mathbf{c}_0$ and $\mathbf{c}_n$ shows that $\mathbf{c}_0^\top \mathbf{v}_0>\mathbf{c}_0^\top\mathbf{v}_i$ for all $i\in [n]$ and $\mathbf{c}_n^\top \mathbf{v}_n>\mathbf{c}_n^\top \mathbf{v}_i$ for all $i\in \{0,\ldots,n-1\}$, thus establishing the statement of \cref{item:A-vertex} in all cases.
This concludes the proof of the lemma.
\end{proof}

We are now all set to finish the construction.
Let $(a_1, \dots, a_n, S, k)$ be an instance of the \textsc{Exact Subset sum with Repetition} problem.
Assume, without loss of generality, that $0\leq a_1 < a_2 < \dots < a_n$.
Let $T$ be the affine transformation we obtain from applying \cref{lemma:replacement-hard-instance} to this instance.
Note that $T(\mathbf{w}_{Ck})$ is the point of $\mathcal{P}\coloneqq T(P_{Ck})$ with the largest $y$ coordinate.
Set $\varepsilon \coloneqq T(\mathbf{w}_{Ck})^y - S$ and note that $0<\varepsilon<\frac{1}{2}(s/a_n)^{Ck/2+1}$ by \cref{lemma:replacement-hard-instance}.
As before we can write $T(\mathbf{x}) = H\mathbf{x} + \mathbf{b}$, for an invertible matrix $H\in \mathbb{R}^{2\times 2}$ and a vector $\mathbf{b}\in \mathbb{R}^2$.
Set the cost vector $\mathbf{c}$ to $(H^{-1})^\top\mathbf{c}_0$.
Note that one can efficiently determine the inverse of $H$, e.g., by using the explicit formula for the inverse of a $2\times 2$ matrix: 
$\begin{psmallmatrix}
 a & b \\ c & d    
\end{psmallmatrix}^{-1} =
\frac{1}{ad-bc}\begin{psmallmatrix}
    d & -b \\ - c & a
\end{psmallmatrix}.
$
In particular, we can determine $\mathbf{c}$ efficiently and its encoding length is polynomial.

Now $T(\mathbf{t}_{Ck})$ is the unique $\mathbf{c}$-maximal vertex of $\mathcal{P}$.
Note that this implies $\mathbf{c}^x<0$ and $\mathbf{c}^y>0$.
Set $\mathbf{u} \coloneqq T(\mathbf{u}_{Ck})$, $\mathbf{w}\coloneqq T(\mathbf{w}_{Ck})$, and $\mathbf{t}\coloneqq T(\mathbf{t}_{Ck})$.
Additionally, let $\mathcal{V}$ be the set of points obtained by applying \cref{lemma:replacement-correct-slopes} to $(a_1,\dots, a_n, S, k)$ and $\mathbf{c}$.
We are now finally ready to define the polygon $P$ that we use in the proof of \cref{thm:monotone-circuit-diameter} as $P\coloneq\mathrm{conv}\left(\{(0,0)^\top, (1,S+\varepsilon)^\top\}\cup \mathcal{V}\cup \mathcal{P}\right)$. 
Set $\mathbf{s} = (0,0)^\top$.

In the following we analyze the properties of $P$.
We begin by explicitly describing all its vertices.
\begin{claim}\label{claim:monotone-construction-vertices}
    The vertices of $P$ are precisely $\mathbf{s}$, $(1, S+\varepsilon)^\top$, the points of $\mathcal{V}$, and the vertices of $\mathcal{P}$.
\end{claim}
\begin{proof}
It directly follows from the definition of $P$ that the set of vertices of $P$ is a subset of the list claimed above. 
Hence, it only remains to show that each of the claimed vertices is indeed a vertex of $P$.
We will certify this by proving that for each point in the list there exists $\mathbf{d}\in \mathbb{R}^2$ such that the respective point is the unique maximizer of the linear functional $\mathbf{d}^\top \mathbf{x}$ in $P$. 

By \cref{lemma:replacement-hard-instance} \ref{item:hard-vertex}, every vertex of $\mathcal{P}$ is the unique maximizer of a linear functional $\mathbf{d}^\top \mathbf{x}$ over $\mathcal{P}$ for some $\mathbf{d}\in \mathbb{R}^2$ with $\mathbf{d}^x < 0$ and $\mathbf{d}^y > 0$. 

Note that for a $\mathbf{d}$ with these properties we have that $\mathbf{d}^\top(0,0)^\top < \mathbf{d}^\top\mathbf{u}$, as $\mathbf{u}^x = 0 = \mathbf{s}^x$ and $\mathbf{u}^y > 0 = \mathbf{s}^y$. 
The same way, $(1, S+\varepsilon)^\top$ has worse value than $\mathbf{w}$, by using $\mathbf{w}^x < 1$ and $\mathbf{w}^y = S +\varepsilon$. 
This implies that also the unique optimizer of $\mathbf{d}^\top \mathbf{x}$ in $\mathcal{P}$ has strictly better $\mathbf{d}$-value than $(0,0)^\top$ and $(1, S+\varepsilon)^\top$.

 Since $\mathbf{d}^x<0, \mathbf{d}^y>0$ and since all points in $\mathcal{V}$ have bigger $x$-coordinate and smaller $y$-coordinate than all points in $\mathcal{P}$, we have that $\mathbf{d}^\top \mathbf{v}\le \mathbf{d}^\top \mathbf{p}$ for every $\mathbf{v}\in \mathcal{V}$ and $\mathbf{p}\in \mathcal{P}$. All in all, this shows that every vertex of $\mathcal{P}$ is indeed also a vertex of $P$.

Next, consider any point $\mathbf{v}_i \in \mathcal{V}$.
Then by \cref{lemma:replacement-correct-slopes} \ref{item:A-vertex} $\mathbf{v}_i$ is the unique maximizer of $\mathbf{d}^\top \mathbf{x}$ in $\mathcal{V}$, for some $\mathbf{d}\in \mathbb{R}^2$ with $\mathbf{d}^x> 0$ and $\mathbf{d}^y < 0$.
For these $\mathbf{d}$, $\mathbf{d}^\top \mathbf{s}$ is smaller than $\mathbf{d}^\top \mathbf{v}_0$, as $\mathbf{s}^x < \mathbf{v}_0^x$ and $\mathbf{s}^y = \mathbf{v}_0^y$.
Similarly, $\mathbf{d}^\top (1,S+\varepsilon)^\top$ is smaller than $\mathbf{d}^\top \mathbf{v}_n$, as $\mathbf{v}_n^x = 1$ and $\mathbf{v}_n^y < S +\varepsilon$.
Furthermore, $\mathbf{d}^x>0, \mathbf{d}^y<0$ and the fact that all points in $\mathcal{V}$ have bigger $x$- and smaller $y$-coordinates than all points in $\mathcal{P}$ implies that we have $\mathbf{d}^\top \mathbf{v}_i\ge \mathbf{d}^\top \mathbf{p}$ for every $\mathbf{p}\in \mathcal{P}$. Summarizing, this shows that indeed all elements of $\mathcal{V}$ are vertices of $P$. 

Finally, we have that $\mathbf{s}=(0,0)^\top$ is the unique maximizer of $(-1,-1)^\top x$ in $P$ and that $(1,S+\varepsilon)^\top$ the unique maximizer of $(1,1)^\top x$ in $P$.
This completes the argument.
\end{proof}

We note that $P$ has a class of canonical non-redundant encodings given by any description of the half-planes defined by two adjacent vertices.
Given two points in the plane one can efficiently determine a polynomially encoded description the line through them.
Thus, one can also determine both half-planes defined by the two points.
We just established the vertices of $P$ in \cref{claim:monotone-construction-vertices}.
Note that the vertices of $\mathcal{P}$ and the points of $\mathcal{V}$ have polynomial encoding length by \cref{lemma:replacement-hard-instance} \ref{item:hard-encoding} and \cref{lemma:replacement-correct-slopes} \ref{item:A-encoding}, respectively.
Furthermore, $\varepsilon$ has polynomial encoding length.
Thus, $P$ has indeed polynomial encoding length and we can compute it efficiently.

In the following, we call the edge between $\mathbf{s}$ and $\mathbf{u}$ the \emph{left edge}, the edge between $\mathbf{s}$ and $\mathbf{v}_0$ the \emph{lower edge}, the edge between $\mathbf{v}_n$ and $(1,S+\varepsilon)$ the \emph{right edge}, and the edge between $\mathbf{w}$ and $(1,S+\varepsilon)^\top$ the \emph{upper edge} of $P$.
We also denote by $\mathcal{T}$ the concave polygonal chain formed by the set of all points in $\mathcal{P}$ that lie on the boundary of $P$.
Next we identify the circuit directions of $P$.

\begin{observation}\label{obs:monotone-construction-circuit-directions}
	The $\mathbf{c}$-monotone circuit directions of $P$ are exactly the positive scalar multiples of the following vectors:
	\begin{enumerate}[label=\textnormal{(\roman*)}]
		\item $(-1,0)^\top, (0,1)^\top$,
		\item $(1,a_i)^\top$ for $i\in [n]$,
		\item the $\mathbf{c}$-monotone edge directions of $\mathcal{T}$.
	\end{enumerate}
    In the following, we will refer to a circuit of $P$ that is a positive scalar of one of the vectors in the first, second, or third item above as a \emph{circuit of type $1$, $2$, or $3$}, respectively.
    Additionally, by \cref{lemma:replacement-hard-instance} \ref{item:hard-slopes}, the circuits of type 3 have a slope of at most $\frac{1}{2Ck}$ in absolute value.
\end{observation}
\begin{proof}
    By \cref{obs:edgesarecircuits} and since we are considering a non-redundant inequality description of $P$, we have that the circuits of $P$ coincide with the vectors parallel to one of its edges. In turn, the edge directions of $P$ are the directions between two consecutive vertices.
    By \cref{claim:monotone-construction-vertices}, together with the coordinates of $\mathbf{u}$, $\mathbf{w}$, $\mathbf{v}_0$, and $\mathbf{v}_n$ specified in \cref{lemma:replacement-hard-instance} \ref{item:hard-coordinates} and \cref{lemma:replacement-correct-slopes} \ref{item:A-coordinates}, respectively, these are $(\pm 1, 0)^\top$, $(0, \pm 1)^\top$, the edge directions between consecutive vertices in $\mathcal{V}=\{\mathbf{v}_0,\ldots,\mathbf{v}_n\}$ (which by \cref{lemma:replacement-correct-slopes} \ref{item:A-slopes} are parallel to $(1,a_i)^\top$ for some $i\in [n]$), and the edge directions of $\mathcal{T}$.

    As $\mathbf{c}^x < 0 < \mathbf{c}^y$, out of the first four only $(-1, 0)^\top$ and $(0,1)^\top$ are $\mathbf{c}$-monotone.
    Next, recall that the slopes of the edges of $\mathcal{T}$ are in $(0, \frac{1}{Ck})$.
    In particular, they are all less than one.
    As such, we must have $|\mathbf{c}^x| < \mathbf{c}^y$.
    This means that $(1, a_i)^\top$ is $\mathbf{c}$-monotone and $(-1,-a_i)^\top$ is not.
\end{proof}

Below we observe that the vertices $\mathbf{v}_i$ are constructed in such a way that they cannot be visited by any $\mathbf{c}$-monotone circuit walk starting in $\mathbf{s}$.
\begin{observation}\label{obs:monotone-construction-lower-part}
A $\mathbf{c}$-monotone circuit walk starting at $\mathbf{s}$ cannot visit any point with a $y$-coordinate of $0$, or any point in $\mathrm{conv}(\mathcal{V})$. In particular, it does not use points on the edge spanned between $\mathbf{v}_{i-1}$ and $\mathbf{v}_i$ for every $i\in [n]$. 
\end{observation}
\begin{proof}
 Recall that $\mathbf{c}^x < 0$.
 Furthermore, we have $\mathbf{p}^x\geq 0$ for any $\mathbf{p}\in P$.
 Thus, we have $\mathbf{c}^\top \mathbf{p} < 0$ for any $\mathbf{p}\in P$ with $\mathbf{p}^y = 0$ and $\mathbf{p}\neq \mathbf{s}$.
 Hence, a $\mathbf{c}$-monotone walk cannot visit $\mathbf{p}$ after starting in $\mathbf{s}$.

 Next we consider the vertices of $\mathcal{V}$.
 By \cref{lemma:replacement-correct-slopes} \ref{item:A-cost} we have $\mathbf{c}^\top \mathbf{v} \leq 0 = \mathbf{c}^\top \mathbf{s}$ for all $\mathbf{v}\in \mathcal{V}$ and thus by convexity also $\mathbf{c}^\top \mathbf{x}\le \mathbf{c}^\top \mathbf{s}$ for all $\mathbf{x}\in \mathrm{conv}(\mathcal{V})$.
 Hence, a $\mathbf{c}$-monotone circuit walk starting at $\mathbf{s}$ cannot visit any of the points in $\mathrm{conv}(\mathcal{V})$. 
\end{proof}

Next, we observe how the distance properties of $P_{Ck}$ carry over to $P$.

\begin{observation} \label{obs:monotone_construction-close-to-t}
    We have $d^P_\mathbf{c}(\mathbf{u}) = d^P_\mathbf{c}(\mathbf{w}) = Ck$.
\end{observation}
\begin{proof}
	For the moment consider any point $\mathbf{p}\in \mathcal{T}$.
    Then making a circuit move in $P$ from $\mathbf{p}$ following a $\mathbf{c}$-monotone circuit direction of type $1$ or $2$ is not feasible, as these directions are steeper than the edges of $\mathcal{T}$.
    Thus, from such points we can only take circuit directions of type $3$.
    
    Furthermore, by construction of $P_{Ck}$, we have that $\mathbf{c}^\top \mathbf{q} = \mathbf{c}^\top \mathbf{u}= \mathbf{c}^\top \mathbf{w}$ for any $\mathbf{q}$ on the line between $\mathbf{u}$ and $\mathbf{w}$.
    This is due to the choice of $\mathbf{c}_0=(1,0)^\top$, $\mathbf{u}_{Ck} = (0,1)^\top$, and $\mathbf{v}_{Ck} = (0,-1)^\top$ and as $T$ is an affine transformation.
    Hence, a $c$-monotone circuit walk starting in $\mathbf{p}$ cannot visit a point outside of $\mathcal{T}$.
    In particular, $T^{-1}$ and $T$ give length-preserving bijections between $\mathbf{c}$-monotone circuit walks starting in $\mathbf{p}$ and $\mathbf{c}_0$-monotone circuit walks starting in $T^{-1}(\mathbf{p})$ in $P_{Ck}$.
	Applying this argument to $\mathbf{u}_{Ck}$ and $\mathbf{w}_{Ck}$ and using \cref{thm:linear-circuit-distance} \ref{item:linear-polygon-distance} we deduce $d^P_\mathbf{c}(\mathbf{u}) = d^P_\mathbf{c}(\mathbf{w}) = Ck$.
\end{proof}

As a final ingredient, we show that points on the upper edge of $P$ that have $\mathbf{c}$-monotone circuit distance to $\mathbf{t}$ less than $Ck$ must be close to $\mathbf{w}$.
This will allow us to show that any short $\mathbf{c}$-monotone circuit walk to $\mathbf{t}$ gives rise to a solution to the subset sum instance.

\begin{lemma} \label{lem:monotone_construction-reaching-P}
	For every point $\mathbf{v}$ with $\mathbf{v}^y = S + \varepsilon$ and $\mathbf{v}^x \geq \frac{1}{2a_n}$ we have $d^P_\mathbf{c}(\mathbf{v}) \geq Ck$.
\end{lemma}
\begin{proof}

For $i\in \{0,\dots, \left\lceil\frac{Ck}{2}\right\rceil\}$ set
$p_i  =
    \left(\frac{s_1}{a_n}\right)^{\left\lceil\frac{Ck}{2}\right\rceil - i + 1}
$.
We will show by induction on $i$ the following statement:
For every point $\mathbf{v}\in P$ with $\mathbf{v}^y = S+\varepsilon$ and $\mathbf{v}^x > p_i$, we have $d^P_\mathbf{c}(\mathbf{v}) \geq 2i + 1$.

Note that for $i = \left\lceil\frac{Ck}{2}\right\rceil$ we have $p_i=\frac{s_1}{a_n}<\frac{1}{2a_n}$, where we used that $s_1\le \frac{1}{2Ck}\le \frac{1}{2}$ by \cref{lemma:replacement-hard-instance},~$(i)$. Hence, the above inductive statement implies the claim of the lemma.
As an induction start, notice that the statement is true for $i = 0$, as $\mathbf{v}^x > p_0$ by \cref{lemma:replacement-hard-instance} implies $\mathbf{v} \notin T(P_{Ck})=\mathcal{P}$ and so in particular $\mathbf{v}\neq \mathbf{t}$ and thus $d^P_\mathbf{c}(\mathbf{v}) \geq 1$.
Hence, assume we proved the statement for some $i\in \{0, \dots, \left\lceil\frac{Ck}{2}\right\rceil-1\}$, and let us show it for $i+1$.

Consider a point $\mathbf{v}\in P$ with $\mathbf{v}^y = S+\varepsilon$ and $\mathbf{v}^x > p_{i+1}$.
Towards a contradiction, suppose that $d^P_\mathbf{c}(\mathbf{v}) \leq 2(i+1)$.
Let $W$ be a $\mathbf{c}$-monotone circuit walk from $\mathbf{v}$ to $\mathbf{t}$ of length at most $2(i+1)$.

Let $\mathbf{q}_1$ denote the successor of $\mathbf{v}$ on $W$.
We cannot have $\mathbf{q}_1 = \mathbf{w}$ as 
\[
d^P_\mathbf{c}(\mathbf{w}) = Ck \geq 2\left(\left\lceil \frac{Ck}{2}\right\rceil-1\right)+1\ge 2i+1\ge d^P_\mathbf{c}(\mathbf{q}_1).
\]
Consider the circuit directions we identified in \cref{obs:monotone-construction-circuit-directions}.
As $\mathbf{q}_1\neq \mathbf{w}$ and $\mathbf{v}^y=S+\varepsilon$, we observe that the circuit step from $\mathbf{v}$  to $\mathbf{q}_1$ along $W$ follows neither a type $1$ nor a type $2$ circuit direction.
Hence, $W$ must use a type 3 circuit direction in this step.
Let $s_j$ be the slope of this circuit direction.
A line through $\mathbf{v}$ with slope $s_j$ intersects the $x=0$ line at $(0, S + \varepsilon - s_j\mathbf{v}^x)$.
Recall that by \cref{lemma:replacement-hard-instance} for all $\bar{\mathbf{v}}\in\mathcal{P}$ we have 
\[\bar{\mathbf{v}}^y\geq S - \frac{1}{2}\left(\frac{s_1}{a_n}\right)^{\left\lceil\frac{Ck}{2}\right\rceil+1}\geq S + \varepsilon - \left(\frac{s_1}{a_n}\right)^{\left\lceil\frac{Ck}{2}\right\rceil+1} > S+\varepsilon-s_jp_{i+1} > S + \varepsilon - s_j\mathbf{v}^x.\] 
Here we used that $\varepsilon\leq\frac{1}{2}\left(\frac{s_1}{a_n}\right)^{\left\lceil\frac{Ck}{2}\right\rceil+1}$ and $s_jp_{i+1}\ge s_1 \left(\frac{s_1}{a_n}\right)^{\left\lceil\frac{Ck}{2}\right\rceil - i}>\left(\frac{s_1}{a_n}\right)^{\left\lceil\frac{Ck}{2}\right\rceil +1}$ in the second and third inequality, respectively.

In particular, we have $(0, S + \varepsilon - s_j\mathbf{v}^x)\in P$ and hence $\mathbf{q}_1 = (0, S + \varepsilon - s_j\mathbf{v}^x)$.
We observe that $W$ has a length of at least two.
Let $\mathbf{q}_2$ denote the point we reach after the first two steps of $W$, i.e., $\mathbf{q}_2$ is the successor of $\mathbf{q}_1$.
We claim that $\mathbf{q}_2^y = S+\varepsilon$ and $\mathbf{q}_2^x > p_i$.
Once we established this claim, using the induction hypothesis we will then be able to conclude $d^P_\mathbf{c}({\mathbf{q}}_2) \geq 2i + 1$ and hence that the length of $W$ is at least $2i+3$, yielding the desired contradiction and concluding the proof of the induction step. 

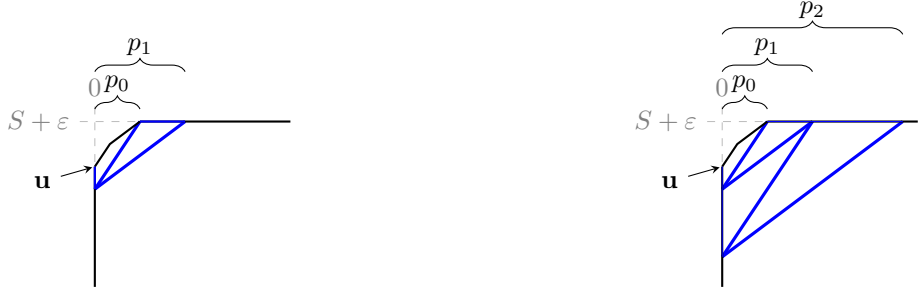
\begin{figure}[ht!]
	\begin{subfigure}[t]{0.5\textwidth}
		\begin{center}
			\begin{tikzpicture}
	\newcommand{\s}{3}
	\newcommand{\e}{0.3}
	\newcommand{\x}{2}

	\begin{scope}
		\coordinate (s) at (0,0);
		\coordinate (u) at (\x,\s + \e);

		\coordinate (v0) at (0   , \s - \e);
		\coordinate (t)  at (0.2, \s      );
		\coordinate (v1) at (0.6 , \s + \e);
	\end{scope}

	\begin{scope}[thin, gray!50!white, dashed]
		\node[left, gray, font=\small] at (-0.2,\s+\e) {$S+\varepsilon$};
		\node[above, gray, font=\small] at (0,\s+.5) {$0$};
		\draw[] (-0.2, \s+\e) -- (0.7,\s+\e);
		\draw[] (0,\s+0.5) -- (0, \s-.5);
	\end{scope}

	\begin{scope}
		\coordinate (s1) at (0.8, \s + \e);
		\coordinate (s2) at (0, \s + \e - 1.2);
		\coordinate (s3) at (1.6, \s + \e);
		\coordinate (s4) at (0, \s + \e - 2.4);

		\coordinate (r1) at (0, \s+\e - 0.9);
		\coordinate (r2) at (1.2, \s + \e);
		\coordinate (r3) at (0, \s+\e - 1.8);
	\end{scope}

	\begin{scope}[red]
\end{scope}
	\begin{scope}[blue,very thick]
		\draw (v1) -- (r1);
		\draw (r1) -- (r2);
\end{scope}

	\begin{scope}[thick]
		\draw (v0) -- (t) -- (v1);
		\draw (v0) -- ($(v0) + (0, -1.6)$);
		\draw (v1) -- ($(v1) + (2.0, 0)$);
	\end{scope}

	\begin{scope}[red, very thick]
\end{scope}

	\begin{scope}[blue, thick]
	\draw (v0) -- (r1);
    \draw (v1) -- (r2);
	\end{scope}

	\begin{scope}[every node/.style={font=\small}]
		\node[xshift = -.7cm, yshift = -.2cm] (alabel) at (v0) {$\mathbf{u}$};
		\draw[-stealth,shorten >=1pt] (alabel) -- (v0);

\draw[decorate,decoration={brace,amplitude=5pt,raise=1ex}] (0, \s + \e) --node[midway, yshift=3ex]{$p_0$} (v1);
        \draw[decorate,decoration={brace,amplitude=5pt,raise=4ex}] (0, \s + \e) --node[midway, yshift=6ex]{$p_1$} (r2);
	\end{scope}

\end{tikzpicture}
 		\end{center}
		\caption{We choose $p_1$ in such a way that starting at $(x, S+\varepsilon)$ with $x>p_1$ we cannot reach a point of $\mathcal{T}$ with two $\mathbf{c}$-monotone circuit moves, unless we visit $\mathbf{u}$ or $\mathbf{w}$.}
	\end{subfigure}
	\begin{subfigure}[t]{0.5\textwidth}
	\begin{center}
		\begin{tikzpicture}
	\newcommand{\s}{3}
	\newcommand{\e}{0.3}
	\newcommand{\x}{2}

	\begin{scope}
		\coordinate (s) at (0,0);
		\coordinate (u) at (\x,\s + \e);

		\coordinate (v0) at (0   , \s - \e);
		\coordinate (t)  at (0.2, \s      );
		\coordinate (v1) at (0.6 , \s + \e);
	\end{scope}

	\begin{scope}[thin, gray!50!white, dashed]
		\node[left, gray, font=\small] at (-0.2,\s+\e) {$S+\varepsilon$};
		\node[above, gray, font=\small] at (0,\s+.5) {$0$};
		\draw[] (-0.2, \s+\e) -- (0.7,\s+\e);
		\draw[] (0,\s+0.5) -- (0, \s-.5);
	\end{scope}

	\begin{scope}
		\coordinate (s1) at (0.8, \s + \e);
		\coordinate (s2) at (0, \s + \e - 1.2);
		\coordinate (s3) at (1.6, \s + \e);
		\coordinate (s4) at (0, \s + \e - 2.4);

		\coordinate (r1) at (0, \s+\e - 0.9);
		\coordinate (r2) at (1.2, \s + \e);
		\coordinate (r3) at (0, \s+\e - 1.8);
        \coordinate (r4) at (2.4,\s + \e);

	\end{scope}

	\begin{scope}[blue, very thick]
	\draw (v1) -- (r1);
	\draw (r1) -- (r2);
	\draw (r2) -- (r3);
	\draw (r3) -- (r4);
	\end{scope}

	\begin{scope}[thick]
		\draw (v0) -- (t) -- (v1);
		\draw (v0) -- ($(v0) + (0, -1.6)$);
		\draw (v1) -- ($(v1) + (2.0, 0)$);
	\end{scope}

	\begin{scope}[blue]
	\draw (v1) -- (r4);
	\draw (v0) -- (r3);
	\end{scope}

	\begin{scope}[every node/.style={font=\small}]
		\node[xshift = -.7cm, yshift = -.2cm] (alabel) at (v0) {$\mathbf{u}$};
		\draw[-stealth,shorten >=1pt] (alabel) -- (v0);

\draw[decorate,decoration={brace,amplitude=5pt,raise=1ex}] (0, \s + \e) --node[midway, yshift=3ex]{$p_0$} (v1);
        \draw[decorate,decoration={brace,amplitude=5pt,raise=4ex}] (0, \s + \e) --node[midway, yshift=6ex]{$p_1$} (r2);
        \draw[decorate,decoration={brace,amplitude=5pt,raise=7ex}] (0, \s + \e) --node[midway, yshift=9ex]{$p_2$} (r4);
	\end{scope}
\end{tikzpicture}
 	\end{center}
	\caption{For general $i\in \{1, \dots,\left\lceil\frac{Ck}{2}\right\rceil]\}$ we choose $p_{i+1}$ in such a way that starting from $(p_{i+1}, S+\varepsilon)$ using the circuit of slope $s_1$ followed by the circuit with slope $a_n$ we reach the point $(p_i, S+\varepsilon)$.}
	\end{subfigure}
	\caption{Visualization of the proof of \cref{lem:monotone_construction-reaching-P}.
    We define the distances $p_i$ in such a way that starting at a point $\mathbf{v}$ with $\mathbf{v}^y = S + \varepsilon$ and $\mathbf{v}^x > p_i$ we can in two $\mathbf{c}$-monotone circuit moves only reach $\mathbf{u}$, $\mathbf{w}$, or points on the upper edge with an $x$-coordinate of at least $p_{i-1}$.
    This is done by analyzing the maximal change in $y$-coordinate and $x$-coordinate, respectively, achievable by the first and second move.
	}
	\label{fig:monotone-construction-upper-left-scaling}
	\end{figure}
So, all that is left is to prove that indeed, $\mathbf{q}_2^y=S+\varepsilon$ and $\mathbf{q}_2^x>p_i$.

As before, $d^P_\mathbf{c}(\mathbf{u})=Ck$ implies that we cannot have $\mathbf{q}_2 = \mathbf{u}$.
Hence, $W$ takes as the second circuit direction a circuit of type $2$ or $3$.
In particular, the slope between $\mathbf{q}_1$ and $\mathbf{q}_2$ is at most $a_n$.
Taking a line from $\mathbf{q}_1$ with slope $a_n$ and intersecting it with the $y = S+\varepsilon$ line yields the point $(\frac{s_j}{a_n}\mathbf{v}^x,S+\varepsilon)$.
As we have $\frac{s_j}{a_n}\mathbf{v}^x > \frac{s_1}{a_n}p_{i+1} \ge \left(\frac{s_1}{a_n}\right)^{\left\lceil\frac{Ck}{2}\right\rceil+1}$, we again have $(\frac{s_j}{a_n}\mathbf{v}^x,S+\varepsilon)\in P$.
Hence $\mathbf{q}_2$ lies between $(\frac{s_j}{a_n}\mathbf{v}^x,S+\varepsilon)$ and $\mathbf{v}$.
In particular, $\mathbf{q}_2$ lies on the upper edge and we have $\mathbf{q}_2^x \geq \frac{s_j}{a_n}\mathbf{v}^x > \frac{s_1}{a_n}p_{i+1} = p_i$.
As discussed above, this finishes the proof.
A visualization of this argument can be found in \cref{fig:monotone-construction-upper-left-scaling}.
\end{proof}

We are now all set to conclude the proof of the theorem.
\begin{proof}[Proof of \cref{thm:monotone-circuit-diameter}]
We will reduce from the \textsc{Exact Subset sum with Repetition} problem, as mentioned before.
Hence, let $(a_1, \dots, a_n, S, k)$ be an instance of the \textsc{Exact Subset sum with Repetition} problem. 
Recall that this means that $k\leq n$ and that any $r\in \mathbb{Z}_{\geq 0}^n$ with $\sum_{i=1}^{n} r_i a_i = S$ must satisfy $\sum_{i=1}^{n} r_i = k$.

Let $P$ be the polygon constructed as above and let $\mathbf{c}$ be the corresponding cost vector.
As argued above we can efficiently construct $P$ and $\mathbf{c}$ and their encoding length is polynomial, establishing \cref{item:main-construct} and \cref{item:main-encoding}.
Furthermore, the edges of $P$ can be classified into the vertical and horizontal edges, the edges on $\mathcal{P}$ and the edges between vertices in $\mathcal{V}$.
In particular, we have $4 + 2Ck + n$ edges, establishing \cref{item:main-edges}.

To finish the proof, we claim that $d^P_\mathbf{c}(\mathbf{s})\leq 2k$ if the subset sum instance has a solution and that $d^P_\mathbf{c}(\mathbf{s}) \geq Ck + 1$ if it does not have a solution.
This shows \cref{item:main-short-walk} and \cref{item:main-long-walk}, respectively.

Let us first show that if the subset sum instance has a solution, then there is a circuit walk of length at most $2k$.
Let $r\in\mathbb{Z}^n$ be a feasible solution for the subset sum instance, i.e., $\sum r_i a_i = S$ and $\sum r_i = k$.
Then we can construct a short circuit walk the following way.
Let $b_1, \dots, b_k$ be an arbitrary order of the $a_i$ in which each $a_i$ appears $r_i$ times.
Starting from $\mathbf{s}$ we alternatingly use a circuit in the direction of $(1, b_j)^\top$ and the circuit in the direction $(-1,0)^\top$.
This gives the following succession of points on the boundary of $P$:
\begin{align*}
	\begin{array}{rcrc}
	&\mathbf{s} & \rightarrow & (1, b_1)^\top \\
	\rightarrow & (0, b_1)^\top & \rightarrow & (1, b_1 + b_2)^\top \\
	& \vdots & & \vdots \\
	\rightarrow & (0, \sum_{j=1}^{k - 1} b_j)^\top & \rightarrow & (1, \sum_{j=1}^{k} b_j)^\top\\
	\rightarrow &  \mathbf{t} &&
		\end{array}
	\end{align*}

Here we use that $\sum_{j=1}^{\ell} b_j = \sum_{i=1}^{n} r_i a_i = S$ and that $\mathbf{t}^y = S$.
Note that by \cref{lemma:replacement-hard-instance} \ref{item:hard-coordinates} all the points in the constructed sequence lie below the upper edge of $P$.
Additionally, all but the last two points lie below $\mathbf{u}$.
Finally, by \cref{lemma:replacement-correct-slopes} \ref{item:A-coordinates} we have $\mathbf{v}^y < 1$ for all $\mathbf{v}\in \mathcal{V}$.
In particular, the points described above lie on the left and right edge of the polygon $P$.
Hence, the chosen circuit directions indeed give rise to the points claimed above and $d^P_\mathbf{c}(\mathbf{s}) \leq 2k$.

For the reverse direction we assume that there is a circuit walk $W$ of length at most $Ck$ from $\mathbf{s}$ to the unique $\mathbf{c}$-optimal vertex $\mathbf{t}$ of $P$, and our goal is to show that then the subset sum instance has a solution. This argument will be divided into two cases, depending on whether or not $W$ uses a point on the upper edge of $P$. Before jumping into those, let us make two useful observations about $W$:

First, recall that the circuit directions of type $3$ have a slope of at most $\frac{1}{2Ck}$.
Since $P\subseteq [0,1]\times \mathbb{R}$, this implies that every step in $W$ following a circuit direction of type $3$ changes the $y$-coordinate by at most $\frac{1}{2Ck}$.
Second, note that using \cref{obs:monotone-construction-lower-part} we know that $W$ does not visit any point in $\mathrm{conv}(\mathcal{V})$ and no points of the lower edge other than $\mathbf{s}$. 

We now proceed with the two cases of the main argument.

\smallskip
\textbf{Case~1. $W$ contains no point from the upper edge of the polygon $P$.}
Recall that by \cref{obs:monotone_construction-close-to-t} we have $d^P_\mathbf{c}(\mathbf{u}) = d^P_\mathbf{c}(\mathbf{w}) = Ck$.
In particular, $W$ cannot contain $\mathbf{u}$.

We next claim that $W$ does not have any step following the circuit direction $(0,1)^\top$.
To see this, note that starting from any point on the left edge of $P$, a circuit move in the circuit direction $(0,1)^\top$ leads to $\mathbf{u}$, which is not visited by $W$.
Similarly, starting from any point on the right edge a circuit move in direction $(0,1)^\top$ leads to a point on the upper edge, that is also not visited by $W$ by assumption.
Finally, it is not feasible to perform a circuit move in direction $(0,1)^\top$ starting from points on $\mathcal{T}$. Altogether, it indeed follows that the circuit direction $(0,1)^\top$ cannot be used by $W$, as claimed.

Let $W'$ be the prefix of $W$ that ends at the first point of $W$ with a $y$-coordinate of at least $S - 0.5$ (this is well-defined, since the last point $\mathbf{t}$ of $W$ has $y$-coordinate equal to $S$).
Now, let $r_i$ for $i\in [n]$ denote the number of times $W'$ uses the circuit direction $(1,a_i)^\top$.

Note that by \cref{lemma:replacement-hard-instance} \ref{item:hard-coordinates} we have $\mathbf{p}^y > S-0.5$ for every point $\mathbf{p}\in \mathcal{P}$.
In particular, performing a circuit move from a point on the left edge below $(0, S-0.5)^\top$ using the circuit direction $(1,a_i)^\top$ changes the $y$-coordinate by precisely $a_i$, or reaches a point on the upper edge.
Thus, the total change in $y$-coordinate in $W$ that stems from circuit directions of type $2$ is $\sum r_i a_i$.
We claim that $\sum r_i a_i = S$.
To show this, observe first that $\sum r_ia_i$ is an integer.
Furthermore, as noted above all steps of $W$ using a circuit direction of type $3$ change the $y$-coordinate by at most $\frac{1}{2Ck}$, so the total change in $y$-coordinate due to such steps is at most $\frac{1}{2}$ in absolute value. Furthermore, as $W$ does not use the circuit direction $(0,1)^\top$, the total change of $y$-coordinate stemming from steps using a circuit-direction of type $1$ equals $0$.

Hence, and as $W'$ reaches a $y$ coordinate of at least $S-0.5$ but not on the upper edge, we must have $\sum r_ia_i \in [S-0.5, S+\varepsilon+0.5]$.
Now, $S$ is the only integer in this range, and so $\sum r_ia_i = S$.
By the assumption on the \textsc{Exact Subset Sum with Repetition} instance, this implies $\sum_{i=1}^{n} r_i = k$ and thus the instance has a solution, as desired. This concludes the proof in the first case.

\smallskip
\textbf{Case~2. $W$ does contain some point from the upper edge of $P$.}
Let $\mathbf{p}$ be the first vertex of $W$ on the upper edge, i.e., $\mathbf{p}^y = S + \epsilon$.
Note that $d^P_\mathbf{c}(\mathbf{p})<Ck$, as witnessed by the suffix of $W$ starting at $\mathbf{p}$.
Thus, applying \cref{lem:monotone_construction-reaching-P}, we have $\mathbf{p}^x \leq \frac{1}{2a_n}$.

Let $\mathbf{q}$ be the predecessor of $\mathbf{p}$ on $W$.
Note that the circuit direction from $\mathbf{q}$ to $\mathbf{p}$ cannot be parallel to $(0,1)^\top$, as otherwise $\mathbf{q}$ would have to be contained in the lower edge of $P$ or lie on the line segment between $\mathbf{v}_{i-1}$ and $\mathbf{v}_i$ for some $i\in [n]$, which is ruled out by \cref{obs:monotone-construction-lower-part}.
The circuit direction from $\mathbf{q}$ to $\mathbf{p}$ also cannot be parallel to $(-1,0)^\top$, since in this case $\mathbf{q}$ would also be contained in the upper edge of $P$, contradicting the definition of $\mathbf{p}$.
Hence, the circuit direction taken from $\mathbf{q}$ to $\mathbf{p}$ is of type $2$ or $3$.
In particular its slope is at most $a_n$.
Therefore using this direction increases the $x$-coordinate by at least $\frac{(S+\varepsilon - \mathbf{q}^y)}{a_n}\ge \frac{(S - \mathbf{q}^y)}{a_n}$. 

As $\mathbf{p}^x \leq \frac{1}{2a_n}$, it follows that $S - \mathbf{q}^y \leq 0.5$ and thus $\mathbf{q}^y\ge S-0.5$.
Let $W''$ be defined as the prefix of $W$ that ends at $\mathbf{q}$.
Then $W''$ reaches the height $\mathbf{q}^y\in [S-0.5, S+0.5]$ and contains no point on the upper edge.
Using the same argument as in the first case, we can construct a feasible solution to the subset sum problem.

Thus, $d^P_\mathbf{c}(\mathbf{s})\leq Ck$ implies that the subset sum instance has a solution, finishing the proof.
\end{proof}

\section{Concluding remarks}
In this work, we focused on monotone circuit walks as these are most directly relevant to circuit augmentation schemes.
However, the non-monotone variant is natural too and states as follows.

\begin{mdframed}[innerleftmargin=0.5em, innertopmargin=0.5em, innerrightmargin=0.5em, innerbottommargin=0.5em, userdefinedwidth=0.95\linewidth, align=center]
	{\textsc{Circuit Distance}}
	\sloppy

	\noindent
	\textbf{Input:} A polytope $P = \{\mathbf{x}\in \mathbb{R}^n\colon Ax\leq \mathbf{b}\}$ defined by a matrix $A\in \mathbb{Q}^{m\times n}$ and a vector $\mathbf{b}\in \mathbb{Q}^m$, two vertices $\mathbf{s}$ and $\mathbf{t}$ of $P$, and $k\in \mathbb{Z}_{\geq 0}$.

	\noindent
	\textbf{Decision:} Is there a circuit walk from $\mathbf{s}$ to $\mathbf{t}$ of length at most $k$?
\end{mdframed}

Our proof techniques seem likely to extend to the undirected setting but require some technical innovation.
On that basis, we conjecture the following:

\begin{conjecture}
\textsc{Circuit Distance} is \NP-hard for polygons.  
\end{conjecture}
\appendix

\section{Missing proofs}\label[appendix]{sec:missing-proofs}
In this section we supply the proofs from previous sections that were left out so as to ease the readability of those parts of the paper.
For context, we repeat the statements.

\begin{repremark}{obs:bruteforce}
For every constant $K\in \mathbb{N}$, there exists a polynomial algorithm that, given as input a polygon $P$ defined by $m$ inequalities, a starting vertex $\mathbf{s}$ of $P$, and a direction $\mathbf{c}\in \mathbb{Q}^2$, computes a $\mathbf{c}$-monotone circuit walk from $\mathbf{s}$ to a $\mathbf{c}$-optimal vertex whose length is at most $\frac{m}{K}$ times the length of a shortest such walk. 
\end{repremark}
\begin{proof}
Let $K\in \mathbb{N}$ be any given constant, and suppose we are given as input a polygon $P$ defined by $m$ inequalities, a starting vertex $\mathbf{s}$ of $P$ and a direction $\mathbf{c}\in \mathbb{Q}^2$ to optimize in.
Then we can compute representatives of all the (at most $m$) equivalence classes of circuits up to scalar multiplication.
Using those, we can then explicitly compute all the (at most $(2m)^K$ possible) circuit walks of length at most $K$ in $P$ starting at $\mathbf{s}$ in time $m^{O(K)}$.
Finally, for each of these circuit walks we can check if they are $\mathbf{c}$-monotone and end in a $\mathbf{c}$-maximal vertex in polynomial time.
Hence, the following algorithm forms a polynomial-time $\frac{m}{K}$-approximation algorithm for the problem of finding a shortest circuit walk from $\mathbf{s}$ to a $\mathbf{c}$-optimal vertex of $P$: If the above procedure finds a monotone circuit-walk of length at most $K$ to a $\mathbf{c}$-optimal vertex, then output the shortest among all such walks.
This then is clearly the optimal solution, with a multiplicative gap of $1$.
Otherwise, the algorithm outputs a monotone edge-walk from $\mathbf{s}$ to a $\mathbf{c}$-optimal vertex of $P$.
This walk clearly has length at most $m$ and is thus no longer than $\frac{m}{K}$ times the length of a shortest $\mathbf{c}$-monotone circuit walk to an optimum, since the latter is bigger than $K$. 
\end{proof}

\begin{reptheorem}{thm:anyfixeddim}
Consider a polygon $P\in \mathbb{R}^2$, a cost vector $\mathbf{c}\in \mathbb{Q}^2$, and a vertex $\mathbf{s}$ of $P$.
For every $d\in \mathbb{Z}_{\geq 2}$ one can efficiently determine a $d$-dimensional polytope $P_d\in \mathbb{R}^d$, a cost vector $\mathbf{c}_d\in \mathbb{Q}$, and a vertex $\mathbf{s}_d$ of $P_d$ such that the following holds:
The length of a shortest $\mathbf{c}$-monotone circuit walk from $\mathbf{s}$ to a $\mathbf{c}$-maximal point of $P$ agrees with the length of a shortest $\mathbf{c}_d$-monotone circuit walk from $\mathbf{s}_d$ to a $\mathbf{c}_d$-maximal point of $P_d$. 
Furthermore, if $P$ has $m$ edges, then one can choose $P_d$ to have $m+d-2$ facets.
\end{reptheorem}
\begin{proof}
Consider the $(d-2)$-dimensional simplex $\Delta_{d-2} = \mathrm{conv}(0,e_{1},\dots,e_{d-2})$, where $e_i$ denotes the $i$-th standard vector.
Then by standard facts about polyhedral products as one may find in \cite{zieglerbook}, the product $P_d = P \times \mathrm{conv}(0,e_{1},\dots,e_{d-2})$ is a $d$-dimensional simple polytope with $m + d-2$ facets.
Furthermore, one can construct this polytope efficiently:
\[P_d = \left\{(\mathbf{x}, \mathbf{y}): \mathbf{x} \in P, \sum_{i=1}^{d-2} \mathbf{y}_{i} \leq 1, \mathbf{y}_{i} \geq 0 \text{ for all } i \in [d-2]\right\}.\]
As noted in Lemma 3.9 of \cite{CircuitDiamConjecture}, the circuits of a product of polytopes $R$ and $Q$ are precisely the vectors $\mathbf{g} \times \mathbf{0}$ or $\mathbf{0} \times \mathbf{h}$ where $\mathbf{g}$ is a circuit of $R$ and $\mathbf{h}$ is a circuit of $Q$. 

Let us first determine the circuits of $\Delta_{d-2}$.
We have $\Delta_{d-2} = \{\mathbf{y}\in \mathbb{R}^{d-2}_{\geq 0}: \sum_{i=1}^{d-2} \mathbf{y}_i  \leq 1\}$.
Consider a $(d-3)\times (d-2)$ sub-matrix $A$ of the matrix defining $\Delta_{d-2}$.
We can have one of two cases.
Either $A$ contains all non-negativity constraints, but the one corresponding to $i$ for $i\in [d-2]$, or $A$ contains all but two non-negativity constraints and the upper bound on the sum.
In the former case, the circuit is a scalar multiple of $e_i$.
In the latter case, let the two missing non-negativity constraints be for the indices $i,j\in [d-2]$ with $i\neq j$.
Then the corresponding circuits are the scalar multiples of $e_i - e_j$.

To finish the proof, set $\mathbf{c}_d = \mathbf{c}\times e_{d-2}$.
Consider the face $P \times e_{d-2}$ and let $\mathbf{x} \times e_{d-2}$ be any point in that face.
Let us consider the different circuit directions of $P_d$.
The first class of circuits have the form $\mathbf{g}\times 0$ for a circuit $\mathbf{g}$ of $P_d$.
Now $\mathbf{g}\times 0$ is feasible at $\mathbf{x}\times e_{d-2}$ if and only if $\mathbf{g}$ is feasible at $\mathbf{x}$.
Furthermore, $\mathbf{g}\times 0$ is $\mathbf{c}_d$-increasing, if and only if $\mathbf{g}$ is $\mathbf{c}$-increasing.
The second class of circuits are of the form $0\times \mathbf{h}$ for a circuit $\mathbf{h}$ of $\Delta_{d-2}$.
Now note that only the circuit directions $\mathbf{x}\times(-e_{d-2})$ and $\mathbf{x}\times (e_i - e_{d-2})$ for $i\in [d-3]$ are feasible at $\mathbf{x}\times e_{d-2}$, but none of these are $\mathbf{c}_d$-improving.
Hence, any $\mathbf{c}_d$-monotone circuit move starting at $\mathbf{x}\times e_{d-2}$ can be seen as a $\mathbf{c}$-monotone circuit move starting at $\mathbf{x}$.

Now set $\mathbf{s}_d\coloneq\mathbf{s}\times e_{d-2}$, which is a vertex of $P_d$.
By the above argumentation and induction, every $\mathbf{c}_d$-monotone circuit walk starting at $\mathbf{s}\times e_{d-2}$ lives in the facet $P\times e_{d-2}$ and only uses directions $\mathbf{g}\times 0$ for $\mathbf{c}$-monotone circuits $\mathbf{g}$ of $P$.
In particular, the $\mathbf{c}_d$-monotone circuit distance from $\mathbf{s}_d$ to a $\mathbf{c}_d$-maximal vertex agrees with the $\mathbf{c}$-monotone circuit distance from $\mathbf{s}$ to a $\mathbf{c}$-maximal vertex of $P$, as desired.
\end{proof}

\begin{repobservation}{obs:edgesarecircuits}
   Let $A \in \mathbb{R}^{m\times 2}$ and $\mathbf{b}\in \mathbb{R}^m$.
   Let $P = \{\mathbf{x} \in \mathbb{R}^2| A\mathbf{x} \leq \mathbf{b}\}$ be a non-empty polygon.
   If no inequality of $A\mathbf{x}\leq \mathbf{b}$ is redundant, then the circuits of $P$ correspond precisely to the vectors parallel to some edge of $P$.
\end{repobservation}
\begin{proof}
For $i\in [m]$ let us denote by $A_i$ the $i$-th row of $A$.
The inequality $A_i\mathbf{x} \leq \mathbf{b}_i$ not being redundant implies that it defines an edge of $P$ for every $i\in [m]$.
As we consider $d=2$, a $(d-1)\times d$-submatrix of $A$ is a single row $A_i$.
Hence, by definition, the circuits of $P$ are precisely the vectors $\mathbf{g}$ with $A_i\mathbf{g}=0$ for some $i\in [m]$.
Now these are precisely the vectors parallel to the edge defined by $A_i\mathbf{x} \leq \mathbf{b}_i$, finishing the proof.
\end{proof}

\begin{reptheorem}{thm:subset-sum-special-hardness}
	The \textsc{Exact Subset sum with Repetition} problem is \NP-hard.
\end{reptheorem}
\begin{proof}
    We will use a slight variation of the standard hardness reduction from \textsc{3-Dimensional Matching} to the \textsc{Subset Sum} problem.
    Recall the \textsc{3-Dimensional Matching} problem, which is \NP-hard as shown by \textcite{karp2010reducibility}.

\begin{mdframed}[innerleftmargin=0.5em, innertopmargin=0.5em, innerrightmargin=0.5em, innerbottommargin=0.5em, userdefinedwidth=0.95\linewidth, align=center]
	{\textsc{3-Dimensional Matching}}
	\sloppy

 \noindent
 \textbf{Input:} Three disjoint sets $X,Y,Z$ of equal size and a subset $E \subseteq X \times Y \times Z$.

 \noindent
 \textbf{Decision:} Is there a subset $M\subseteq E$ such that every element of $ X, Y$, and $Z$ is part of exactly one element of $M$?
\end{mdframed}
An output $M$ with the desired property is also referred to as a \emph{perfect matching}.

    Given an instance of \textsc{3D Matching}, we want to construct an instance of the \textsc{Exact Subset Sum with Repetition} problem that is feasible if and only if there exists a perfect matching for the \textsc{3D Matching} instance.
    Take arbitrary orderings $X = \{x_0, \dots, x_{N-1}\}$, $Y = \{y_0, \dots, y_{N-1}\}$, and $Z = \{z_0, \dots, z_{N-1}\}$.
    We will define numbers with base $B := N + 1$ with $3N+1$ digits.
    Here digit $i$ for $0 \leq i \leq N-1$ corresponds to $x_i$.
    Digit $j$ with $N \leq j \leq 2N-1$ corresponds to $y_{j-N}$.
    And finally a digit $k$ with $2N \leq h \leq 3N-1$ corresponds to $z_{h-2N}$.
    Finally, there is an additional (most significant) digit at position $3N$ which we will use to ensure that the sequence of numbers we construct is a feasible input of \textsc{Exact subset sum with repetition}.

    We set the target for our instance of \textsc{Exact subset sum with repetition} as $S:=NB^{3N} + \sum_{\ell = 0}^{3N-1} B^\ell$, i.e., the number in base $B$ represented by a value of $N$ in position $3N$ followed by ones at all other positions.
    Additionally set $k:=N$.
    For each element $e = (x_i, y_j, z_h) \in E$ we add the number $B^i + B^{j+N} + B^{h+2N} + B^{3N}$ to our instance, i.e., the number has a one precisely in positions $i, j+N, h+2N, 3N$ and zeroes elsewhere.
    Let $a_1,\dots, a_m$ with $m:=|E|$ be the list of numbers created in this way, and note that these numbers are pairwise distinct. 
    The \textsc{3D Matching} instance is trivially infeasible if $|E| < N$.
    Hence, \textsc{3D Matching} remains \NP-hard when restricted to instances such that $|E|\ge |N|$, and thus we can assume that the subset sum instance consists of $|E| \geq N = k$ elements in the following, meeting one of the requirements on the input for \textsc{Exact subset sum with repetition}.

    Note that for any $r\in \mathbb{Z}_{\geq 0}^m$ we have 
    \[
        \sum_{i=1}^m r_i a_i \geq B^{3N}\sum_{i=1}^m r_i\enspace .
    \]
    As $S < (N+1)B^{3N}$, we have $\sum_{i=1}^m r_i a_i > S$ whenever $\sum_{i=1}^m r_i > N = k$.
    Furthermore, by definition of $k$ and $B$, the sum $\sum_{i=1}^m r_i a_i$ has no carry-over between bits when $\sum_{i=1}^m r_i \leq k$.
    Hence, if we have $\sum_{i=1}^m r_i a_i = S$, then every $r_i$ must be $0$ or $1$.
    Additionally, for every position $d\in [3N-1]$ there must be exactly one $i$ with $r_i = 1$ such that $a_i$ contains the summand $B^d$.
    Hence, taking $M$ to be the elements of $E$ corresponding to the $a_i$ with $r_i = 1$ gives rise to a perfect matching forming a solution to the \textsc{$3$D matching} instance.
    In particular, as any such perfect matching must contain precisely $N$ elements, it follows that $\sum_{i=1}^m r_i = N = k$. Thus, $(a_1,\ldots,a_m,S)$ form a feasible instance of the \textsc{Exact subset sum with repetition} problem, as desired.

    For the other direction of the desired equivalence, it suffices to note that every perfect matching for the \textsc{3D Matching} instance gives rise to a solution of the \textsc{Exact Subset sum with Repetition} instance by setting $r_i$ to one for precisely those indices $i$ where $a_i$ corresponds to an element of the perfect matching.
\end{proof}

\begin{repobservation}{obs:transform}
	Let $P = \{\mathbf{x}\in \mathbb{R}^2\colon A\mathbf{x}\leq \mathbf{b}\}$ be a polygon defined by $A\in \mathbb{Q}^{m\times 2}$ and $\mathbf{b}\in \mathbb{Q}^m$.
  Consider an affine transformation defined by an invertible matrix $H\in \mathbb{Q}^{2\times 2}$ and a translation vector $d\in \mathbb{Q}^2$.
	Let $W = (\mathbf{x}_1, \dots, \mathbf{x}_n)$ be a circuit walk in $P$. Then $W':=(H\mathbf{x}_1 + \mathbf{d}, \dots, H\mathbf{x}_n + \mathbf{d})$ is a circuit walk in the transformed polytope $HP + \mathbf{d} = \{\mathbf{x}\in \mathbb{R}^2\colon AH^{-1}\mathbf{x} \leq \mathbf{b} + AH^{-1}\mathbf{d}\}$. Furthermore, if $W$ is $\mathbf{c}$-monotone for some $\mathbf{c}\in \mathbb{R}^2$, then $W'$ is $\mathbf{c}'$-monotone for $\mathbf{c}'\coloneqq (H^\top)^{-1}\mathbf{c}$.
\end{repobservation}
\begin{proof}
By definition, a vector $\mathbf{g}\in \mathbb{R}^2\setminus\{0\}$ is a circuit of $P$ if and only if there exists some $i\in [m]$ such that the $i$-th row $A_i$ of $A$ is nonzero and satisfies $A_i\mathbf{g}=\mathbf{0}$. But since the latter is clearly equivalent to the $i$-th row $(AH^{-1})_i=A_iH^{-1}$ of $AH^{-1}$ being non-zero and satisfying $(AH^{-1})_i(H\mathbf{g})=0$, we can see that $\mathbf{g}$ is a circuit of $P$ if and only if $H\mathbf{g}$ is a circuit of $HP+\mathbf{d}$. 

It follows directly from this that for every circuit walk $(\mathbf{x}_1,\ldots,\mathbf{x}_n)$ in $P$ and every $i\in [n]$, the point $H\mathbf{x}_i+\mathbf{d}$ is obtained from $H\mathbf{x}_{i-1}+\mathbf{d}$ by a circuit move in $HP+\mathbf{d}$ along the direction of the circuit $H(\mathbf{x}_i-\mathbf{x}_{i-1})$ of $HP+\mathbf{d}$. This shows that the transformed walk  $W'=(H\mathbf{x}_1 + \mathbf{d}, \dots, H\mathbf{x}_n + \mathbf{d})$ is a circuit walk in $HP + \mathbf{d}$, as desired.

Finally, assume that $W$ is $\mathbf{c}$-monotone for some $\mathbf{c}\in \mathbb{R}^2$. Let $\mathbf{c}':=(H^\top)^{-1}\mathbf{c}$. Then we have, for every $i\in [n-1]$:
$$(\mathbf{c}')^\top((H\mathbf{x}_i+\mathbf{d})-(H\mathbf{x}_{i-1}+\mathbf{d}))=\mathbf{c}^\top H^{-1}H(\mathbf{x}_{i}-\mathbf{x}_{i-1})=\mathbf{c}^\top\mathbf{x}_i-\mathbf{c}^\top\mathbf{x}_{i-1}>0,$$ showing that $W'$ is indeed $\mathbf{c}'$-monotone. This concludes the proof. 
\end{proof}

\printbibliography

\end{document}